\documentclass[11pt, letterpaper]{article}

\usepackage{comment}
\usepackage{fullpage}
\usepackage{microtype}
\usepackage{titling}
\usepackage[utf8]{inputenc}
\usepackage{amsmath,amsthm,algorithm,appendix,calc,enumitem,mathtools}
\usepackage{mathrsfs}
\usepackage{graphicx}
\usepackage{comment}
\usepackage{titling}
\usepackage{physics}
\usepackage{relsize}

\usepackage{thm-restate}

\usepackage[pdftex, pagebackref, colorlinks = true, linktocpage=true]{hyperref}
\hypersetup{colorlinks,
    linkcolor=blue, 
    citecolor=ForestGreen,      
    urlcolor=red,    
}
\usepackage[nameinlink]{cleveref}

\usepackage[dvipsnames]{xcolor}
\usepackage{enumitem}
\usepackage{todonotes}

\usepackage{prettyref}

\usepackage[compact]{titlesec}

\usepackage[T1]{fontenc}
\usepackage[american]{babel}

\usepackage{float}
\usepackage{mathpazo}

\usepackage[varg]{pxfonts} %
\usepackage[full]{textcomp}
\usepackage[scr=rsfso]{mathalfa}%
\usepackage{bm} %
\usepackage[T1]{fontenc}

\DeclareMathOperator*{\E}{\mathlarger{\mathbb{E}}}

\newcommand{\Z}{\mathbb{Z}}

\newcommand{\R}{\mathbb{R}}
\newcommand{\N}{\mathbb{N}}

\newcommand{\cH}{\mathop{\mathcal{H}}}
\renewcommand{\bar}{\overline}
\renewcommand{\epsilon}{\varepsilon}
\newcommand{\eps}{\varepsilon}

\usepackage{nicefrac}

\newcommand{\handout}[5]{
   \renewcommand{\thepage}{#1-\arabic{page}}
   \noindent
   \begin{center}
   \framebox{
      \vbox{
    \hbox to 5.78in { {\bf #1}
     	 \hfill #2 }
       \vspace{4mm}
       \hbox to 5.78in { {\Large \hfill #5  \hfill} }
       \vspace{2mm}
       \hbox to 5.78in { {\it #3 \hfill #4} }
      }
   }
   \end{center}
   \vspace*{4mm}
}

\newtheorem{ctr}{Counter}[section]

\newtheorem{theorem}[ctr]{Theorem}
\newtheorem{corollary}[ctr]{Corollary}
\newtheorem{lemma}[ctr]{Lemma}

\newtheorem{proposition}[ctr]{Proposition}
\newtheorem{definition}[ctr]{Definition}

\newtheorem{fact}[ctr]{Fact}

\newtheorem{conjecture}[ctr]{Conjecture}

\newtheorem{remark}[ctr]{Remark}

\newcommand{\iprod}[1]{\left\langle{#1}\right\rangle}

\newenvironment{proof-sketch}{\noindent{\bf Sketch of Proof:}\hspace*{1em}}{\qed\bigskip}
\newenvironment{proof-idea}{\noindent{\bf Proof Idea}\hspace*{1em}}{\qed\bigskip}
\newenvironment{proof-of-lemma}[1]{\noindent{\bf Proof of Lemma #1}\hspace*{1em}}{\qed\bigskip}
\newenvironment{proof-attempt}{\noindent{\bf Proof Attempt}\hspace*{1em}}{\qed\bigskip}

\makeatletter
\def\fnum@figure{{\bf Figure \thefigure}}
\def\fnum@table{{\bf Table \thetable}}
\long\def\@mycaption#1[#2]#3{\addcontentsline{\csname
  ext@#1\endcsname}{#1}{\protect\numberline{\csname
  the#1\endcsname}{\ignorespaces #2}}\par
  \begingroup
    \@parboxrestore
    \small
    \@makecaption{\csname fnum@#1\endcsname}{\ignorespaces #3}\par
  \endgroup}
\def\mycaption{\refstepcounter\@captype \@dblarg{\@mycaption\@captype}}
\makeatother

\newcommand{\mathify}[1]{\ifmmode{#1}\else\mbox{$#1$}\fi}
\newcommand{\bigO}O

\newcommand{\remove}[1]{}
\newcommand{\ignore}[1]{}

\def\Z{\mathbb{Z}}

\def\R{\mathbb{R}}

\def\cA{{\cal A}}

\def\cE{{\cal E}}

\def\cG{{\cal G}}
\def\cH{{\cal H}}
\def\cI{{\cal I}}

\def\cK{{\cal K}}
\def\cL{{\cal L}}

\def\cN{{\cal N}}

\def\cP{{\cal P}}
\def\cQ{{\cal Q}}
\def\cR{{\cal R}}

\def\cU{{\cal U}}

\def\cX{{\cal X}}

\providecommand{\norm}[1]{\lVert #1 \rVert}

\renewcommand{\bar}{\overline}
\renewcommand{\epsilon}{\varepsilon}
\newcommand{\gam}{\gamma}
\newcommand{\sig}{\sigma}

\newcommand{\lam}{\lambda}
\newcommand{\Lam}{\Lambda}
\newcommand{\el}{\ell}
\newcommand{\al}{\alpha}

\newcommand{\unif}{\in_{\text{R}}}
\DeclareMathOperator{\Pois}{Pois}
\DeclareMathOperator{\Bin}{Bin}
\DeclareMathOperator{\Ber}{Ber}
\DeclareMathOperator{\sgn}{sgn}
\newcommand{\sg}{\ensuremath{\mathsf{SG}}}
\newcommand{\sgxi}{\ensuremath{\mathsf{SG}_\xi}}
\newcommand{\csp}{\ensuremath{\mathsf{CSP}_\Lam(\alpha)}}

\newcommand{\cspdens}[1]{\ensuremath{\mathsf{CSP}_\Lam({#1})}}

\DeclareMathOperator{\supp}{supp}

\newcommand{\toas}{\overset{\textnormal{a.s.}}{\to}}

\newcommand{\eqas}{\overset{\textnormal{a.s.}}{=}}

\DeclareMathSizes{10.95}{10.95}{8}{7.5}

\usepackage{parskip}

\DeclareMathOperator{\Stab}{Stab}

\definecolor{niceish}{HTML}{1cb5bd} 

\author{
   Chris Jones\thanks{
     Department of Computer Science, University of Chicago, Chicago, Illinois, USA. Supported by the National Science Foundation under Grant No. CCF-2008920.
     Email: \textcolor{red}{csj@uchicago.edu}.
   }\and 
   Kunal Marwaha\thanks{
     Department of Computer Science, University of Chicago, Chicago, Illinois, USA. 
     Supported by the National Science Foundation Graduate Research Fellowship Program under Grant No. DGE-1746045. 
     Email: \textcolor{red}{kmarw@uchicago.edu}.
   }\and 
   Juspreet Singh Sandhu\thanks{
     School of Engineering \& Applied Sciences, Harvard University, Cambridge, Massachusetts, USA.
    Supported by a Simons Investigator Fellowship, NSF grant DMS-2134157, NSF STAQ award PHY-1818914, DARPA grant W911NF2010021, DARPA ONISQ program award HR001120C0068 and DOE grant DE-SC0022199.
     Email: \textcolor{red}{jus065@g.harvard.edu}.
   }\and 
   Jonathan Shi\thanks{
    Department of Computer Science, Bocconi University, Milan, Italy.
    Supported by European Research Council (ERC) award No. 834861 (SO-ReCoDi).
     Email: \textcolor{red}{jonathan.shi@unibocconi.it}.
   }
 }
\title{Random Max-CSPs Inherit Algorithmic Hardness from Spin Glasses}
\date{\today}
\usepackage[normalem]{ulem} 

\begin{document}
\maketitle
\begin{abstract}
We study random constraint satisfaction problems (CSPs) at large clause density.
We relate the structure of near-optimal solutions for \emph{any} Boolean Max-CSP to that for an associated spin glass on the hypercube, using the Guerra-Toninelli interpolation from statistical physics.
The \emph{noise stability} polynomial of the CSP's predicate is, up to a constant, the mixture polynomial of the associated spin glass.
We show two main consequences:
\begin{enumerate}
    \item We prove that the maximum fraction of constraints that can be satisfied in a random Max-CSP at large clause density is determined by the ground state energy density of the corresponding spin glass.
    Since the latter value can be computed with the Parisi formula~\cite{parisi1980sequence, talagrand2006parisi,auffinger2017parisi}, 
    we provide numerical values for some popular CSPs.
    \item We prove that a Max-CSP at large clause density
    possesses generalized versions of the \emph{overlap gap property}
    if and only if the same holds for the corresponding spin glass. 
    We transfer results from \cite{huang2021tight} to obstruct 
    algorithms with \emph{overlap concentration} on a large class of Max-CSPs. This immediately includes local classical and local quantum algorithms \cite{chou2022limitations}.
\end{enumerate}
\end{abstract} 
\pagenumbering{roman}

\thispagestyle{empty}
\newpage

\thispagestyle{empty}
\setcounter{tocdepth}{2}
{
    \hypersetup{linkcolor=blue}
    \tableofcontents
    \thispagestyle{empty}
}
\thispagestyle{empty}
\newpage

\pagenumbering{arabic}

\section{Introduction}

In this work, we formalize a general and deep connection between two intensely studied classes of optimization problems: constraint satisfaction problems (CSPs), studied in computer science, and spin glass models, studied in statistical physics.
We demonstrate that as the clause density of the random CSP increases, the geometric properties of the set of nearly-optimal solutions converge to those of a corresponding spin glass model.
In these spin glass models, the very same geometric properties imply bounds on the average-case approximability achieved by broad classes of algorithms~\cite{alaoui2020optimization, gamarnik2020low, gamarnik2021overlap}; these bounds are conjectured to be the best possible among all polynomial-time algorithms~\cite{gamarnik2021topological, huang2021tight}.
The correspondence we establish here implies that the same lower bounds apply to average-case CSPs.

CSPs are paradigmatic computational tasks.
Their study has led to foundational results in computational hardness, approximability, and optimization~\cite{paschos2013paradigms}.
In recent years, we have learned more about CSPs through methods inspired by statistical physics, especially when the clauses
of the CSP are chosen randomly~\cite{ding2014satisfiability, ding2016maximum, dembo2017extremal, sen2018optimization, DingSlySun2022}. By identifying the solution quality of a variable assignment with the \emph{energy} of a configuration of particles, we can investigate ``physical''
properties of the CSP, such as phase transitions or solution clustering at different temperatures.
Surprisingly, these physical properties can have computational consequences.

We study random CSP instances with Boolean variables, random literal signs, and number
of constraints which is a large constant times the number of variables, such that each constraint acts on a constant number of variables.
If $n$ is the number of variables, then $m = \alpha n$ is the number of constraints for some constant $\alpha$. 
For large enough $\alpha$, the CSP is unsatisfiable with probability $1 - o_n(1)$.
Therefore the goal
is to find a variable assignment that maximizes the number of satisfied constraints,
and we think of these as \emph{Max-CSPs}.

Given a random Max-CSP, how many constraints can be satisfied? How are the best assignments distributed around the hypercube? Can we find these assignments with efficient algorithms?
Statistical physicists use questions like these to investigate the \emph{solution geometry} of a problem. 
Our main result connects the solution geometry of a Max-CSP (with large enough $\alpha$) to that of a \emph{spin glass}.
As a consequence, much of our mathematical and algorithmic understanding of spin glasses
transfers to CSPs at large clause density.

A \emph{spin glass} (more properly a \emph{mixed mean-field spin glass}) is a random system of $n$ particles (variables) specified by a \emph{mixture polynomial} $\xi(s) = \sum_{p \ge 1} c_p^2 s^p$.
In this model, the interaction strength between every $p$-tuple of particles is an independent Gaussian with variance $c_p^2 n^{1-p}$; this can be thought of as a randomly-weighted CSP on the complete $p$-uniform hypergraph.
We show that as the clause density of any random CSP increases ($\al \to \infty$),
the solution space starts to resemble that of a spin glass.
For example, Max-Cut on random graphs with large constant average degree qualitatively looks like the \emph{Sherrington-Kirkpatrick model}, where $c_2 = 1, \{c_i\}_{i \neq 2} = 0$~\cite{sherrington1975solvable, dembo2017extremal}. (Note that our proof applies to Max-2XOR instead of Max-Cut, so we do not recover this result exactly.)

\subsection{Main results}

Formally, we relate the \emph{free energy density} of a random Max-CSP instance to that of a particular spin glass.
The associated spin glass is determined only by the Fourier weights of the CSP. In fact, the mixture polynomial of the spin glass is, up to a constant, the \emph{noise stability polynomial} of the CSP (see \Cref{sec:notation} for definitions):
\begin{theorem}[Free energy density]
\label{thm:main_informal}
Generate a random CSP instance $\cI$ (with cost function $H_\cI$) consisting of $\alpha \cdot n$ independent and uniform constraints
of a predicate $f : \{\pm 1\}^k \to \{0,1\}$ with randomly signed literals.
Define the polynomial
\begin{align}
\label{eq:main-eqn}
\xi(s) 
=  \Stab_s(f) - \widehat{f}(\varnothing)^2 
= \sum_{j = 1}^k  \|f^{=j}\|^2 s^j\,,
\end{align}
where $\Stab_s(f)$ is the noise stability polynomial of $f$, and $\|f^{=j}\|^2$ is the Fourier weight of $f$ at degree $j$.
Generate a random spin glass instance $H^{\xi}$ with mixture polynomial $\xi$. 

Let $\beta > 0$, and define $Z_{\cI}(\beta)=\sum_{\sigma\in\{\pm 1\}^n} H_{\cI}(\sigma)$ and $Z_{\sgxi{}}(\beta)=\sum_{\sigma\in\{\pm 1\}^n} H^{\xi}(\sigma)$ as the respective partition functions. Then
\begin{align}
\frac{1}{\beta n}\log Z_{\mathcal{I}}(\beta)  = \widehat{f}(\varnothing) +
    \frac{1}{\beta n}\frac{\log Z_{\sgxi{}}(\beta)}{\sqrt{\al}}
    + O\left(\tfrac{\beta^2}{\alpha^{2}}\right) +
    o_n\left(1\right)\,.
\end{align}
where the second-to-last term (which may depend on $n$) satisfies $\abs{O\left(\tfrac{\beta^2}{\alpha^{2}}\right)} \leq C \cdot \tfrac{\beta^2}{\al^2}$ whenever $\tfrac{\beta}{\al} \leq \eps_0$ for absolute constants $C, \eps_0 >0$,
and the last term is random and is 
$o_n\left(1\right)$ w.h.p.
\end{theorem}

Prior work relates the free energy density of specific CSPs such as Max-$k$XOR \cite{dembo2017extremal, sen2018optimization} and Max-$k$SAT \cite{panchenko2018ksat} to that of a spin glass. \Cref{thm:main_informal} generalizes this connection to any random Max-CSP with randomly signed literals.

The asymptotic equivalence of the free energy density implies the equivalence of several properties of
the solution geometry for large enough $\alpha$.
We show two specific implications of \Cref{thm:main_informal}. The first is that the optimal value of a random Max-CSP in the large clause density limit can be found with a spin glass calculation. 

\begin{restatable}{corollary}{optimumValue}
\label{cor:maximum_value}
\textnormal{(Optimal value equivalence).}
Generate a random CSP instance $\cI$ consisting of $\alpha \cdot n$ independent and uniform constraints
of a predicate $f : \{\pm 1\}^k \to \{0,1\}$ with randomly signed literals.
Let $v_{\cI}$ be the maximum fraction of constraints of $\cI$ that can be satisfied.
Let $\xi$ be defined as in \Cref{eq:main-eqn}, and $GSED(\sgxi{})$ as the (non-random) ground state energy density of the associated spin glass. Then
\begin{align}\label{eq:sat-fraction}
v_{\cI} = \widehat{f}(\varnothing) +
\frac{GSED(\sgxi{})}{\sqrt{\alpha}}
+ o\left(\tfrac{1}{\sqrt{\alpha}}\right)\,.
\end{align}
where the last term is at most $\kappa(n, \al)$  w.h.p for a function $\kappa(n, \al)$ satisfying $\kappa(n,\al)\cdot \sqrt{\al} \to 0$ as $n \to \infty$ then $\al \to \infty$.
\end{restatable}

Computing the minimum value, or ground state energy, of a spin glass can famously
be done using the \emph{Parisi formula}.
In \Cref{sec:optimalvalue}, we use the Parisi formula to compute $GSED(\sgxi{})$ for several common CSPs (our code is available online).

The second implication relates to algorithmic hardness.\footnote{For this implication, we need to boost \Cref{thm:main_informal} to interpolate the free energy density restricted to any given overlap and correlation structure.} 
Intuitively, when global minima are located in clusters, some algorithms cannot efficiently find them.
A recent body of work (starting with \cite{gamarnik2014limits}) has pinpointed hardness on the presence of an \emph{overlap gap property (OGP)}, showing how this property obstructs noise-stable algorithms from $(1-\eps)$-approximating
average-case instances.
We show that a flexible version of the OGP, called the \emph{branching OGP}~\cite{huang2021tight}, exists on a spin glass exactly when it exists on the associated Max-CSP at large enough clause density. 
As a result, the techniques that obstruct algorithms on certain spin glasses also work on the corresponding Max-CSPs.

\begin{corollary}[OGP equivalence, informal]
\label{cor:ogp}
Take any Max-CSP. Consider the associated spin glass \sgxi{}, where $\xi$ is defined as in \Cref{eq:main-eqn}. Then \sgxi{} exhibits an OGP at value $v$ if and only if the Max-CSP exhibits an OGP at value $\widehat{f}(\varnothing) + \frac{v}{\sqrt{\alpha}}$ for all sufficiently large $\alpha$.
\end{corollary}

The formal class of algorithms we obstruct is as follows~\cite{huang2021tight}.
Consider two correlated instances with correlation parameter $t \in [0,1]$.
A deterministic\footnote{A randomized algorithm is overlap-concentrated if it is overlap-concentrated for every fixing of the randomness.} algorithm is \emph{overlap-concentrated} if for every $t$, the overlap (a.k.a the Hamming distance) between the two output assignments for the two correlated instances is concentrated inside a narrow interval.
This is a notion of noise robustness that
many commonly-used algorithms have, including  approximately Lipschitz algorithms on spin glasses~\cite{huang2021tight} and local classical and local quantum algorithms on random Max-CSPs~\cite{chou2022limitations}.
Additionally, survey propagation with a constant number of message-passing rounds likely has this property~\cite{braunstein2005survey, gamarnik2021topological, bresler2022algorithmic}.

\begin{theorem}[OGP obstructs algorithms with overlap concentration, informal]
\label{cor:ogp_and_overlap_concentration_obstructed}
Consider an algorithm $\cA$ with \emph{overlap concentration}. Then $\cA$ cannot output arbitrarily good approximate solutions on instances of Max-CSPs which exhibit an OGP.
\end{theorem}

Our proof of \Cref{cor:ogp_and_overlap_concentration_obstructed} follows that of Huang and Sellke~\cite{huang2021tight}, who prove the same result on spin glasses. Only a few changes are needed to transfer their result to Max-CSPs.
Note that this result also applies to CSPs with small $\al$ if they exhibit an OGP.

It is known that an OGP (specifically, a branching OGP) exists with high probability on spin glasses with even mixture polynomials without quadratic terms~\cite{chen2019suboptimality, huang2021tight}. Combining this with the results that we have stated, we conclude the following:
\begin{corollary}[informal]
\label{cor:obstruct-even-csps}
Consider a random Max-CSP with a predicate $f$ such that the only nonzero Fourier coefficients of $f$ have even degree $j \ge 4$. 
For almost all instances of the Max-CSP, no algorithm with overlap concentration can output $(1-\epsilon)$-approximately optimal solutions for all $\epsilon > 0$.
\end{corollary}

For example, this unconditionally obstructs overlap concentrated algorithms from approximating random 4XOR instances.

\Cref{cor:ogp} makes progress on~\cite[Problem 9.2 (arXiv version)]{chou2022limitations},
asking which CSPs exhibit an OGP at large clause density.
A full characterization of spin glasses with an OGP, and thereby CSPs in the large $\al$ limit, is not complete (although \emph{spherical spin glasses} have been classified~\cite[Proposition 1]{subag18}).
For example, the Sherrington-Kirkpatrick model is strongly suspected to have no OGP,
but this is not fully proven~\cite{auffinger2020sk}.

We cannot help but mention that these topological properties of the solution space (i.e. the solution
geometry) may \emph{precisely characterize} the algorithmic approximability of spin glasses and random CSPs~\cite{achlioptas2008algorithmic}.
For spin glasses, the following is known:
\begin{theorem}
\label{thm:alg_precise_characterize}
For all spin glasses \sgxi{}, there is a value $ALG$ (given by an extended Parisi formula) such that:
\begin{enumerate}
    \item \cite{montanari2021optimization, alaoui2020optimization, sel21} 
    Assume that the minimizer of the extended Parisi functional for \sgxi{} exists. Then for all $\eps > 0$, there is an efficient algorithm that outputs a solution with value $ALG - \eps$ on almost all instances.
    \item \cite{huang2021tight} Assume the mixture polynomial $\xi$ is \emph{even}. For all $\eps > 0$, the spin glass exhibits a branching OGP with value $ALG + \eps$, which therefore obstructs overlap-concentrated algorithms from achieving this value on almost all instances.
\end{enumerate}
\end{theorem}

It is likely the case that the same holds for CSPs with sufficiently large $\alpha$.
Specifically, there are explicitly computable constants $ALG \leq OPT$ for any CSP predicate $f$ such that 
the optimum value of the CSP is about $\widehat{f}(\emptyset) + \frac{OPT}{\sqrt{\al}}$~(\Cref{cor:maximum_value}),
but the threshold for efficient algorithms appears to be $\widehat{f}(\emptyset) + \frac{ALG}{\sqrt{\al}}$.
In this paper we prove that the lower bound transfers (Part 2).
The upper bound (Part 1) is known for Max-Cut by~\cite{alaoui2021local},
and can likely be generalized to an arbitrary CSP (in a similar way that~\cite{alaoui2020optimization} generalizes~\cite{montanari2021optimization}
for spin glasses).
Existence of an OGP may also characterize approximability of random CSPs for small $\al$,
but proving this is significantly more difficult.

The paper is organized as follows.
We provide formal definitions and additional motivation in \Cref{sec:notation}.
In \Cref{sec:mainproof}, we give the main interpolation between every Max-CSP and a related 
spin glass.
In \Cref{sec:optimalvalue}, we show how this implies equivalence of optimal value~(\Cref{cor:maximum_value}), and list numerical approximations to optimal values of several common Max-CSPs.
In \Cref{sec:ogp_from_freeenergy}, we prove that the main interpolation implies the equivalence of OGPs~(\Cref{cor:ogp}). 
In \Cref{sec:hardness}, we prove that an OGP obstructs overlap-concentrated algorithms on Max-CSPs (\Cref{cor:ogp_and_overlap_concentration_obstructed}) and conclude \Cref{cor:obstruct-even-csps}.
We close with a discussion in \Cref{sec:discussion}. 
The appendices contain some technical proofs. 

\subsection{Related work}

\paragraph*{Spin glasses.}
Spin glasses as a state of matter have been studied since the early 20th century.
They were first considered as metallic alloys with many ground states, in which the magnetic spins of the individual particles in the alloy
are frustrated (i.e. many nearby spins are mismatched).
In the language of CSPs,
spin glasses have
exponentially many near-optimal solutions, in which many constraints are unsatisfied.
In this work, we reserve the term ``spin glass'' for \emph{mean-field} spin glass models,
where ``mean-field'' means that all pairwise (or higher arity) interactions between particles are present.
The spins (i.e. the possible values for each variable) are Boolean (also called \emph{Ising}) for all spin glasses that we consider,
although in other physical and mathematical settings they may be $[q]$-valued or vector-valued.

The Sherrington-Kirkpatrick model
is an early mathematical model of a spin glass~\cite{sherrington1975solvable}.
It was solved (by deriving an explicit formula for the free energy density) by Parisi~\cite{parisi1980sequence}, but the solution relied on non-mathematically rigorous physical arguments. A later series of works \cite{guerra2003broken, talagrand2006parisi, panchenko2013sherrington} proved that the formula is correct for all mixed spin glass models.
Its numerical value was carefully approximated for the Sherrington-Kirkpatrick model \cite{Crisanti_2002} and more recently for other mixed spin glasses~\cite{alaouialgorithmic, marwaha_kxor}.

\paragraph*{Random CSPs through the lens of statistical physics.}
Statistical physicists have studied random instances of combinatorial optimization problems since at least the 1980s~\cite{Fu_1986, mezard1987spin}.
A partial ``dictionary'' converting between the language of computer science and physics
is provided in~\cite[Table 1]{chou2022limitations}.

Several modern works use the Guerra-Toninelli interpolation~\cite{guerra2004high} in a similar technical way as our work; the interpolation method is by now a standard tool in spin glass theory.
Dembo, Montanari, and Sen~\cite{dembo2017extremal} applied the interpolation to prove that the size of the Max-Cut and Max-Bisection in a random $d$-regular or Erd\H{o}s-R{\'e}nyi graph is related to the Sherrington-Kirkpatrick model in the same way as \Cref{cor:maximum_value}.
The interpolation was later used (with different spin glass models)
to determine the optimal value of random Max-$k$SAT~\cite{panchenko2018ksat}, Max-$k$XOR and Max $q$-cut~\cite{sen2018optimization} instances in the highly unsatisfiable regime.
Compared to these works, our \Cref{thm:main_informal} generalizes the CSP predicate to arbitrary mixtures of Boolean functions, but they must have random signs on the literals (e.g. Max-$2$XOR instead of Max-Cut; this is used in the debiasing argument (\Cref{app:debiasing}) and to simplify a technical part of the proof (\Cref{eqn:whyweneedrandomsigns}). We are not certain which of our results extend without random signs.
We also extend the Guerra-Toninelli interpolation (in \Cref{thm:overlap-geometry}) to transfer more properties of the solution geometry than just the optimal value; this is exactly what allows us to compare results on algorithmic hardness.

In this paper, we study the \emph{highly unsatisfiable} regime, where the number of clauses of the CSP is $\alpha n$ for some large constant $\alpha$. 
When the clause density $\alpha$ is smaller, for example near the satisfiability threshold of the CSP, the exact connection with spin glasses breaks down, and existing results are less unified.
Nonetheless, methods inspired by statistical physics continue to give powerful insight into the solution structure of these CSPs~\cite{achlioptas2002asymptotic,achlioptas2006solution,panchenko2004bounds,ding2014satisfiability, ding2016maximum, DingSlySun2022}.

\paragraph*{Overlap gaps.}
When near-optimal solutions are clustered, it becomes impossible for many algorithms to find them.
From a geometric perspective, we can't ``move'' from one cluster to others without passing through a lower-value assignment. This general phenomenon was named the \emph{overlap gap property (OGP)} and it was shown to obstruct local algorithms~\cite{gamarnik2014limits}.
Further generalizations of overlaps
\cite{chen2019suboptimality, huang2021tight} show stronger obstructions on wider classes of algorithms.
The OGP and its generalizations have been used to obstruct algorithms from finding near-optimal solutions of various quantities in mixed spin glasses~\cite{gamarnik2020low, gamarnik2021overlap, huang2021tight, sel21}, CSPs~\cite{chen2019suboptimality, bresler2022algorithmic, chou2022limitations}, sparse random graphs~\cite{gamarnik2014limits, rahman2017local}, and matrices~\cite{gamarnik2018finding, ogp_pca}. 
See \cite{gamarnik2021topological, huang2022computational} for a survey of overlaps and solution geometry.

\paragraph*{Optimizing spin glasses and random Max-CSPs.}
Recently, \cite{alaoui2020optimization, sel21} showed that a type of approximate message-passing algorithm finds the ground state energy of spin glasses without overlap gaps.
Under the same assumption, a version of this algorithm was also shown to be optimal on Max-Cut for sparse random graphs with constant (but sufficiently large) degree \cite{alaoui2021local}. 
Both the Sherrington-Kirkpatrick model and Max-Cut on sparse random graphs are strongly suspected to have no overlap gaps~\cite{auffinger2020sk}.

There has been some study of a near-term quantum algorithm (the QAOA \cite{farhi2014quantum}) optimizing spin glasses, with recently-proven rigorous performance bounds~\cite{Claes_2021, Farhi2022quantumapproximate}. In fact, for large enough clause density, the performance of the QAOA is \emph{identical} on a random instance of Max-$k$XOR and on its corresponding spin glass~\cite{boulebnane2021, basso2021quantum, zhou_qaoa_hyper}.

\section{Preliminaries}
\label{sec:notation}
\subsection{Random constraint satisfaction problems (CSPs)}

A CSP instance, denoted $\cI$,
consists of $n$ variables
and a set of constraints, denoted $E(\cI)$.
In a random CSP instance, the constraints $E(\cI)$
are drawn from a distribution $\Lambda$ on functions $f : \Sigma^k \to \R$.

\begin{definition}[Instance of a random CSP]
\label{defn:CSP}
Let $\Lambda$ be a distribution on functions $f : \Sigma^k \to [-1,1]$ with alphabet $\Sigma = \{\pm 1\}$. Fix a constant $\alpha > 0$. Then, a random CSP over $n$ variables $\{\sigma_i\}_{i=1}^n$
and $m = \alpha n$ clauses is generated by:
for each $i = 1, \dots, m$, draw $i_1, \dots, i_k$ uniformly i.i.d from $[n]$,
    draw $f \sim \Lambda$, draw $k$ random signs $\eps_i$ uniformly i.i.d from $\{\pm 1\}$,
    and add the constraint $e$ to $E(\cI)$, describing the clause $f_e(\sigma_e) := f(\eps_1 \sigma_{i_1}, \dots, \eps_k \sigma_{i_k})$.
\end{definition}

\begin{remark}\label{rmk:random-signs}
    The canonical case is to take $\Lam$ which is supported on a single predicate $f : \Sigma^k \to \{0,1\}$.
    For example, the OR predicate corresponds to $k$SAT.
    Our proofs apply to the more general setting of \Cref{defn:CSP}, but
    it does not apply to $|\Sigma| > 2$ or instances without random signs.
    Note that the constraints are scaled so that the output of $f$ is always in the interval $[-1,+1]$.
    Predicates of mixed arity may be simulated using $k$-ary predicates that ignore some input variables. 
\end{remark}
We denote this random model as $\csp$. We study the \emph{highly unsatisfiable} regime; for example, one can think of $\alpha \gg \exp(k)$.

In this work, we use the language of \emph{Hamiltonians}.
In other settings this quantity may be called the objective function, the value, or the score of the assignment.

\begin{definition}[CSP Hamiltonian]
\label{defn:csp_hamiltonian}
Consider a $\csp$ instance $\cI$ on $[n]$. For any input $\sigma \in \{\pm 1\}^n$ let
\begin{align}
 H^{\alpha}(\sigma) = \frac{1}{\alpha}\sum_{e \in E(\cI)}f_e(\sigma_e)\,.
\end{align}
\end{definition}

\begin{remark}
    We divide by $\alpha$ so that, regardless of the value of $\alpha$, the value of the Hamiltonian 
    is in the same interval $[-n, +n]$.
\end{remark}
\begin{definition}[Optimal value of a CSP instance]
Define the maximum or optimal value of a $\csp$ instance $\cI$ by:
\begin{align}
    v_{\cI} = \frac{1}{n}\max_{\sigma \in \{\pm 1\}^n} H^{\alpha}(\sigma)\,.
\end{align}
\end{definition}
When the predicates in the CSP are 0/1-valued, the maximum value of a CSP instance $\mathcal{I}$ is the maximum possible fraction of constraints that can be satisfied.

\begin{remark}
By a union bound, with high probability all assignments to a random CSP with predicate $f$
satisfy $\widehat{f}(\varnothing) + O(\tfrac{1}{\sqrt{\alpha}})$ fraction of constraints.
The purpose of our work is to precisely study the behavior of the $\frac{1}{\sqrt{\alpha}}$ term.
\end{remark}

\begin{remark}
    The constraints of an instance of $\csp{}$ are slightly dependent
    because of the fixed number of edges $m = \alpha n$.
    In the proof we will pass to the ``Poisson model'' in which the number of edges is $\Pois(\al n)$, or equivalently,
    the multiplicity of each constraint is an i.i.d $\Pois\left(\tfrac{\al}{2^k n^{k-1}}\right)$ random variable.
\end{remark}

\subsection{Mean-field spin glasses}

Let $n$ denote the number of particles or variables in the system. We introduce the mathematical objects describing a spin glass:
\begin{definition}[Mixture polynomial]
A \emph{mixture polynomial} is $\xi(s) = \sum_{p \ge 1}^k c_p^2 s^p$ for some nonnegative coefficients $c_p^2$.
\end{definition}

\begin{definition}[Gaussian disorder]
The \emph{disorder coefficients} $J^{(p)}$ are an order-$p$ tensor where each axis has length $n$. When $J^{(p)}$ is a Gaussian tensor, we say that there is \emph{Gaussian disorder}. Specifically, for each $(i_1, \dots, i_p) \in [n]^p$, we have $J^{(p)}[i_1, \dots, i_p] \sim \cN(0,1)$ i.i.d.\footnote{We follow the definition in~\cite[Chapter 2]{panchenko2013sherrington}. 
Some definitions of Gaussian disorder use a \emph{symmetric} Gaussian tensor, which is equivalent up to scaling of the $c_k$ and $o_n(1)$ change in the free energy density. This is because the models only differ on tensor entries $[i_1, \dots, i_p]$ such that $i_1, \dots, i_p$ are not all distinct, which only make up $o_n(1)$ fraction of all entries
of the tensor.}
\end{definition}

We now have enough to define our spin glasses.
The mixture polynomial specifies the random model, and the disorder coefficients determine the instance of the model. 
For example, $\xi(s) = s^2$ (or $\xi(s) = s^2/2$ in some works) specifies the well-studied \emph{Sherrington-Kirkpatrick spin glass model}.

\begin{definition}[Finite mixed spin glass]
Fix a mixture polynomial $\xi$.
A \emph{spin glass} is a random Hamiltonian $H^\xi :\{\pm 1\}^n \to \R$ given by
sampling Gaussian disorder $J := (J^{(p)} : p = 1,2,\dots)$, where the $J^{(p)}$ are independent,
then 
\begin{align}
H^{\xi}(\sigma) = \sum_{p = 1}^k \frac{c_p}{n^{(p-1)/2}} \iprod{J^{(p)}, \sigma^{\otimes p}} = \sum_{p = 1}^k \frac{c_p}{n^{(p-1)/2}} \sum_{(i_1,\dots,i_p) \in [n]^p} J^{(p)}[i_1,\dots,i_p] \sigma_{i_1} \dots \sigma_{i_p}\,.
\end{align}
We say that $H^\xi$ is an \emph{instance} of the random model \sgxi{}.
\end{definition}

$\sgxi{}$ is also known as the \emph{mixed $p$-spin model}\footnote{The $p = 1$ term is a Gaussian external field. Some spin glass models are instead defined with a \emph{fixed} external field.} (in contrast to the pure $p$-spin model, in which the mixture polynomial is $s^p$, i.e. all interactions have the same size $p$).

An early motivator of spin glass theory was the conjecture (and eventual proof) of the following representation for the ground state energy of the model:

\begin{theorem}[Parisi formula \cite{parisi1980sequence, guerra2003broken, talagrand2006parisi, panchenkomixed, auffinger2017parisi}]
\label{thm:parisi-formula}
Fix any spin glass model \sgxi{}. 
Sample a sequence of independent instances $\{H^{\xi}_{n=1}, H^\xi_{n=2},...\}$ from \sgxi{} for increasing values of $n$. Then the limit
\begin{align}
    GSED(\sgxi{}) := \lim_{n\to\infty} \max_{\sigma\in\{\pm 1\}^n} \frac{1}{n}H^{\xi}_{n}(\sigma)\, ,
\end{align}
almost surely exists. Furthermore,
\begin{align}
    GSED(\sgxi{}) &\eqas
    \min_{\zeta \in \cU}\cP^{\xi}_{\infty}(\zeta)\,,
\end{align}
where $\cP^{\xi}_{\infty}$ is the Parisi functional at zero temperature with mixture polynomial $\xi$.
The Parisi functional at zero temperature and the set $\cU$ are defined in \Cref{defn:parisivariational}.
\end{theorem}

We state one other fact about spin glasses. 
One way to characterize a spin glass is as a Gaussian process on the state space $\{\pm 1\}^n$ with covariance structure given by the mixture polynomial $\xi$. The covariance of the Gaussians $H^\xi(\sigma_1), H^\xi(\sigma_2)$ is a function of the normalized inner product, or \emph{overlap}, of $\sigma_1$ and $\sigma_2$.
\begin{fact}[Gaussian characterization of a spin glass]\label{fact:cov-gaussian-process}
    Let $\{H^{\xi}(\sigma)\}_{\sigma \in \{\pm1\}^n}$ be a spin glass instance of the random model \sgxi{}. Then these variables are jointly Gaussian, with mean zero and covariance
    \begin{align}
         \E\left[H^{\xi}(\sigma_1)H^{\xi}(\sigma_2)\right] =   n\cdot\xi\left(\frac{\langle \sigma_1, \sigma_2\rangle}{n}\right)\, ,
    \end{align}
    where $\langle\cdot\,, \cdot\rangle$ is the standard inner product.
\end{fact}

\subsection{Solution geometry of optimization problems}
\label{sub:solution-geometry}

\emph{Solution geometry} refers to the distribution of the optimal and near-optimal solutions on the Boolean hypercube. 
Given a Hamiltonian $H : \{\pm 1\}^n \to \R$, what do we mean by ``near-optimal solutions''? We consider two notions.
The first, used when considering computational tasks, is the set of $\sigma$ such that $H(\sigma) \geq v$ for some value $v$.
The second, a ``smoother'' notion from statistical physics,
is that near-optimal solutions are samples from a low-temperature Gibbs distribution.

\begin{definition}[Gibbs distribution]
Given a Hamiltonian $H: \{\pm 1 \}^n \to \R$ and an \emph{inverse temperature} parameter $\beta \geq 0$, the \emph{Gibbs distribution} is the distribution
on $\{\pm 1\}^n$ with probability proportional to
$e^{\beta H(\cdot)}$.
\end{definition}
\begin{definition}[Partition function]
The 
\emph{partition function} of $H$ is the normalizing constant
of the Gibbs distribution, i.e. the exponentially-weighted sum
\begin{align}
    Z_H(\beta) = \sum_{\sigma \in \{\pm1\}^n} e^{\beta H(\sigma)}\,.
\end{align}
\end{definition}

The most important geometric notion is the \emph{overlap} of two assignments, which is exactly the normalized inner product and is linearly proportional to the Hamming distance between the assignments:
\begin{definition}[Overlap]
\label{defn:normalizedoverlap}
    The \emph{overlap} between two assignments (configurations) $\sigma_1, \sigma_2 \in \{\pm 1\}^n$ is the normalized inner product,
    \begin{align}
        R(\sigma_1,\sigma_2) = \frac{\langle \sigma_1, \sigma_2 \rangle}{n} \, .
    \end{align}
\end{definition}
Note that the overlap is always in $[-1,1]$.
We study the distribution of $R(\sigma_1, \sigma_2)$ when $\sigma_1, \sigma_2$ are independently sampled near-optimal solutions.
Roughly speaking, a gap in the support of this distributions (i.e. an \emph{overlap gap}) implies obstructions for some types of algorithms.

We also study pairwise overlaps between more than two samples. An example is the distribution of the 3-tuple of overlaps in a ``triangle'' of independent near-optimal solutions $\sigma_1, \sigma_2, \sigma_3$.
Topological gaps in this distribution will also obstruct algorithms.
The generalization of overlap used in this paper is the \emph{$I$-overlap}; given a set $\bm{\sigma}$ of $|\bm{\sigma}| = \ell$ assignments, each in $\{\pm 1\}^n$, we define an overlap value for every subset of $\bm{\sigma}$.

\begin{definition}[$I$-overlap]
\label{defn:multi-overlap}
Consider any $\bm{\sigma} := ({\sigma}_1, \dots, {\sigma}_{\ell}) \in (\{\pm 1\}^n)^{\ell}$ and any $I \subseteq [\ell]$. Define the \emph{$I$-overlap} of $\bm{\sigma}$ as
    \begin{align}
        R_I(\bm{\sigma}) = \frac{1}{n} \sum_{j=1}^n \prod_{i \in I} ({\sigma}_i)_j\,.
    \end{align}
\end{definition}
The $I$-overlap recovers the definition of overlap when $I$ has two elements. 

We define an \emph{overlap vector} to associate a set of $\ell$ assignments with all of its possible $I$-overlaps:
\begin{definition}[Overlap vector]
\label{defn:overlap-vector}
Consider any $\bm{\sigma} := ({\sigma}_1, \dots, {\sigma}_{\ell}) \in (\{\pm 1\}^n)^{\ell}$. Then its overlap vector $Q(\bm{\sigma}) \in [-1,1]^{2^{[\ell]}}$ lists all possible $I$-overlaps; that is, $Q(\bm{\sigma})_I = R_I(\bm{\sigma})$ for every $I \in 2^{[\ell]}$. 
\end{definition}

\begin{definition}[Overlap polytope]
Consider the set of all possible overlap vectors $Q(\bm{\sigma})$ for any positive $n \in \mathbb{N}$ and sets of vectors $\bm{\sigma} \in (\{\pm 1\}^n)^\ell$. Then $\cR^{(\ell)} \subseteq [-1,1]^{2^{[\ell]}}$ is the closure of this set. Formally,
\begin{align}
    \cR^{(\ell)} = \overline{\{Q(\bm{\sigma}): n \in \N \mathrel{,} \bm{\sigma} \in (\{\pm 1\}^n)^{\ell}\}}\,.
\end{align}

\end{definition}
\begin{remark}
    Since $R_\varnothing(\bm{\sigma}) = 0$ for all $\bm{\sigma}$, we can ignore this coordinate,
    and then $\cR^{(\ell)}$ is a non-degenerate convex polytope in $\R^{2^{\el}-1}$.
    For example, $\cR^{(2)}$ is a regular tetrahedron in $\R^3$.
\end{remark}

\begin{definition}[Preimage of overlap vectors]
For any open subset $S \subseteq \cR^{(\ell)}$ (in the Euclidean subset topology of\ $\cR^{(\el)}$), let $U^{(\ell)}_n(S) \subseteq (\{\pm 1\}^{n})^\ell$ be the preimage of $S$ in dimension $n$; i.e. the set of $\bm{\sigma} \in (\{\pm 1\}^{n})^\ell$ such that $Q(\bm{\sigma}) \in S$.
We drop the superscript when $\ell$ is in context.
\end{definition}

We are interested in the distribution on $\cR^{(\el)}$
of the overlap vector $Q(\sigma_1, \dots, \sigma_\el)$ when $\sigma_1, \dots, \sigma_\el$ are independently sampled near-optimal solutions.\footnote{It is likely that the distribution converges in some sense as $n \to \infty$, and then converges again as $\beta \to \infty$, but our proof does not show this nor use this.}
Our main result shows that the support of this distribution for a Max-CSP in the $n \to \infty$ then $\al \to \infty$ limit equals that of a corresponding spin glass.
Hence, \emph{overlap gaps} (formally defined in \Cref{sec:ogp_from_freeenergy}) transfer between the two models.

The proof uses the concept of \emph{free energy}, defined as
the logarithm of the partition function.

\begin{definition}[Free energy]
The \emph{free energy} of $H$ is $\log Z_H(\beta)$.
\end{definition}

The free energy at low temperature ($\beta \to \infty$) matches the optimal value of the Hamiltonian. This is a commonly-used fact (for example, \cite[Equation 1.7]{panchenko2013sherrington}); we include the short proof.
\begin{fact}
\label{fact:maxvalue}
The maximum value of a Hamiltonian $H: \Omega \to \R$ on a finite domain $\Omega$ is related to its free energy by
\begin{align}
    \max_{\sigma \in \Omega} H(\sigma)
    \le 
     \frac{1}{\beta}
     \log Z_H(\beta)
     \le 
      \max_{\sigma \in \Omega} H(\sigma) + \frac{\log |\Omega|}{\beta}\,.
\end{align}
\end{fact}
\begin{proof}
We have 
\begin{align}
     \frac{1}{\beta}
     \log \sum_{\sigma \in \Omega}
     \exp
     \left(
     \beta
     H(\sigma)
     \right)
     \ge 
          \max_{\sigma \in \Omega} \frac{1}{\beta}
     \log
     \exp
     \left(
     \beta
     H(\sigma)
     \right) = \max_{\sigma \in \Omega} H(\sigma)\,,
\end{align}
and 
\begin{align}
      \frac{1}{\beta}
     \log \sum_{\sigma \in \Omega}
     \exp
     \left(
     \beta
     H(\sigma)
     \right)
     \le 
      \max_{\sigma \in \Omega} 
         \frac{1}{\beta} 
     \log |\Omega|
     \exp
     \left(
     \beta
     H(\sigma)
     \right)
     =  \max_{\sigma \in \Omega} H(\sigma) + \frac{\log |\Omega|}{\beta}\,.
\end{align}
\end{proof}

\subsubsection{Extended remark on multi-overlaps}

The $I$-overlap for $|I| > 2$ is a type of \emph{multi-overlap}, i.e. a higher-order version of the standard overlap.
In fact, the structures we will use to obstruct algorithms only constrain the pairwise overlaps with $|I| = 2$.
The pairwise overlaps characterize functions that are rotationally-invariant.
Multi-overlaps may be useful when studying CSPs at small $\alpha$; see the concluding remarks in \Cref{sec:discussion}.

The $I$-overlap is general enough
to capture any property of the solution geometry which is \emph{permutation-invariant} under the action of $S_n$ permuting the bits of $\{\pm 1\}^n$.
Formally, we have the following:
\begin{definition}[Permutation-invariant function]\label{defn:perm-inv}
    Let a permutation $\pi \in S_n$ act on $\sigma \in \{\pm 1\}^n$ as $\sigma^{\pi}_i = \sigma_{\pi(i)}$
    and on a function $f : \{\pm 1\}^n \to \R$ (extended coordinate-wise to $f: (\{\pm 1\}^n)^\el \to \R$) as
        $f^\pi(\sigma) = f(\sigma^\pi)$.
    Then $f$ is permutation-invariant if it is fixed by all $\pi$.
\end{definition}
\begin{fact}
Any permutation-invariant function $f: (\{\pm 1\}^n)^\el \to \mathbb{R}$ can be expressed as a function on $\cR^{(\el)}$,
\begin{align}
    f(\sigma_1, \dots, \sigma_\el) = f(Q(\sigma_1, \dots, \sigma_\el))\,.
\end{align}
\end{fact}

If $f(\sigma_1, \dots, \sigma_\el)$ is additionally invariant under the action of $S_\el$ which permutes the inputs then $f$ is determined by the
overlaps $R_{\{1\}}, R_{\{1,2\}}, R_{\{1,2,3\}}, \dots, R_{\{1,2,3,\dots, \el\}}$.
The corresponding overlap polytope is $[-1,+1]^\el$, which is significantly smaller and simpler.
These overlaps are a more common definition of multi-overlap than \Cref{defn:multi-overlap} \cite{barbier2022strong}.

\subsection{Fourier analysis on the hypercube}
\label{sub:fourier_analysis_prelim}

Fourier analysis is commonly used to study Boolean functions~\cite{o2014analysis}.
For example, the Fourier basis provides a convenient way to understand the action of linear operators on a Boolean function.

We consider the space of Boolean functions $f: \{\pm 1\}^k \to \R$ with the expectation inner product, over the uniform distribution on $\{\pm 1\}^n$.
These functions have a canonical decomposition.

\begin{definition}[Fourier spectrum of a Boolean function]
    Every function $f: \{\pm 1\}^k \to \R$ permits a unique decomposition as a linear combination of parity functions. Specifically,
    \begin{align}
        f(\sigma) = \sum_{S \subseteq [k]}\widehat{f}(S)\prod_{i \in S}\sigma_i\, ,
    \end{align}
    where $\widehat{f}(S)$ are called the Fourier coefficients.
\end{definition}
One can verify that the monomials $\left\{\prod_{i \in S}\sigma_i\right\}_{S \subseteq [k]}$ form an orthonormal basis, often called the \emph{Fourier basis}.

Recall that the average value of $f$ is exactly the Fourier coefficient of the empty set:
\begin{fact}
$\E_{\sigma \in \{\pm 1\}^n} [f(\sigma)] = \widehat{f}(\varnothing)$.
\end{fact}

We also consider the \emph{noise stability} of a Boolean function, which describes how resistant the function is to independent noise on its input bits.
\begin{definition}[$\rho$-correlated]\label{defn:rho-correlated}
Fix $\sigma \in \{\pm 1\}^k$ and $\rho \in [-1,+1]$. Then a random sample $\tau$ of the distribution $N_\rho(\sigma)$ chooses each spin $\tau_i$ independently as
    \begin{align}
        \tau_i = \begin{cases}
        \sigma_i & \text{with probability } \frac{1+\rho}{2}\,, \\
        -\sigma_i & \text{with probability }\frac{1-\rho}{2}\,.
        \end{cases}
        \end{align}
    Since $\E_{\tau \sim N_\rho(\sigma)}[\tau_i \sigma_i] = \rho$, we say that $\tau$ is $\rho$-correlated with $\sigma$.
\end{definition}
\begin{definition}[Noise stability of a Boolean function around a point]
\label{defn:noise-stability-point}
    For $\rho \in [-1,+1]$, define the noise stability of the Boolean function $f: \{\pm 1\}^k \to \R$ around point $\sigma \in \{\pm1\}^n$ as
    \begin{align}
        \Stab_\rho[f](\sigma) = \E_{\tau \sim N_\rho(\sigma)}\bigg[f(\sigma)f(\tau)\bigg]\,.
    \end{align}

\end{definition}

\begin{definition}[Noise stability of a Boolean function]\label{defn:noise-stability}
The noise stability of a Boolean function $f: \{\pm1\}^k \to \R$ is defined as
\begin{align}
    \Stab_\rho[f] =
    \E_{\sigma \sim \{\pm 1\}^k}
    \left[
    \E_{\tau \sim N_\rho(\sigma)}
    \left[f(\sigma)f(\tau)\right]
    \right] = \E_{\sigma \sim \{\pm1\}^k}\left[\Stab_\rho[f](\sigma) \right] \,.
\end{align}
\end{definition}

\begin{fact}[{\cite[Theorem 2.49]{o2014analysis}}]
\begin{align}
\Stab_\rho[f] = \displaystyle\sum_{S \subseteq [k]} \rho^{|S|} \widehat{f}(S)^2
= \sum_{j = 0}^k \rho^j\|f^{=j}\|^2\,,
\end{align}
where $\|f^{=j}\|^2 = \sum_{T \subseteq [k], |T| = j} \widehat{f}(T)^2$
is the Fourier weight of the $j$th Fourier level of $f$.
\end{fact}

\subsection{Concentration inequalities}
\label{sec:concentration}

In the probabilistic combinatorics literature, ``with high probability'' means with probability $1-o_n(1)$,
where $f = o_n(g)$ denotes $\lim_{n \to \infty}\tfrac{f(n)}{g(n)} = 0$.

\begin{remark}\label{rmk:concentration}
All of the quantities that we consider are exponentially concentrated. As a consequence, for all $\eps > 0$, the stated theorems hold with $\eps$ slack with probability at least $1-\exp(-C n)$, for some $C >0$ depending on $\eps$.
For ease of exposition, we have opted to simply state the results with high probability.
\end{remark}

\begin{lemma}[{Concentration under spin glass Gibbs distribution, \cite[Theorem 1.2]{panchenko2013sherrington}}]
\label{lem:gaussian-concentration}
    Let $\beta > 0$ and ${f : \{\pm1\}^n\to \R_{\geq 0}}$ be arbitrary. Let $\{H(\sigma)\}_{\sigma \in \{\pm1\}^n}$ be a Gaussian process such that $\E H(\sigma)^2 \leq s^2$ for all $\sigma \in \{\pm1\}^n$.
    Let $X = \log \sum_{\sigma \in \{\pm1\}^n} \exp(\beta H(\sigma))f(\sigma).$
    Then for all $x \geq 0$,
    \begin{equation}
            \Pr[\abs{X - \E X} \geq x] \leq 2 \exp\left(\frac{-x^2}{4s^2 \beta^2}\right)\,.
    \end{equation}
\end{lemma}
For a spin glass $H^\xi$, we will think of $X$ as having mean $C \cdot n + o(n)$, which formally holds if $f$ satisfies $\exp(-C' \cdot n) \leq f(\sigma) \leq \exp(C' \cdot n)$ for all $\sigma$. By \Cref{fact:cov-gaussian-process}, $\E H^\xi(\sigma)^2 = n \cdot \xi(1) = O(n)$, so $X$ typically fluctuates by only $O(\sqrt{n})$.
An important special case of this concentration inequality is concentration of
the maximum value of $H$. The proof is in \Cref{app:csp-technical}.

\begin{restatable}{corollary}{maxConcentration}
\label{cor:max-concentration}
    Let $\{H(\sigma)\}_{\sigma \in \{\pm1\}^n}$ be a Gaussian process such that $\E H(\sigma)^2 \leq s^2$ for all $\sigma \in \{\pm1\}^n$.
    Then for all $x \geq 0$,
    \begin{equation}
            \Pr[\abs{\max_{\sigma \in \{\pm 1\}^n} H(\sigma) - \E \max_{\sigma \in \{\pm 1\}^n} H(\sigma)} \geq x] \leq 2 \exp\left(\frac{-x^2}{4s^2}\right)\,.
    \end{equation}
\end{restatable}

We use a similar concentration inequality for the CSP setting.
The proof is found in \Cref{app:csp-technical}.
\begin{restatable}[Concentration under CSP Gibbs distribution]{lemma}{cspConcentration}\label{prop:poisson-concentration}
    Let $\beta > 0$ and $f : \{\pm 1\}^n \to \R_{\geq 0}$ be arbitrary. Let $\{H(\sigma)\}_{\sigma \in \{\pm 1\}^n}$ be a sample from \csp{}. Let 
    $X = \log \sum_{\sigma \in \{\pm1\}^n} \exp(\beta H(\sigma))f(\sigma)$.
    Then for all $x \geq 0$,
    \begin{equation}
            \Pr[\abs{X - \E X} \geq x] \leq 2 \exp\left(\frac{-x^2 \alpha}{4(n + \tfrac{x}{\beta})\beta^2}\right)\,.
    \end{equation}
\end{restatable}

Exactly analogously to \Cref{cor:max-concentration}, we conclude concentration of the max.

\begin{corollary}\label{cor:max-concentration-csp}
    Let $\{H(\sigma)\}_{\sigma \in \{\pm1\}^n}$ be a sample from \csp{}.
    Then for all $x \geq 0$,
    \begin{equation}
            \Pr[\abs{\max_{\sigma \in \{\pm 1\}^n} H(\sigma) - \E \max_{\sigma \in \{\pm 1\}^n} H(\sigma)} \geq x] \leq 2 \exp\left(\frac{-x^2\al}{4(n+x)}\right)\,.
    \end{equation}
\end{corollary}

\section{Sparse and dense models have the same free energy}
\label{sec:mainproof}
In this section we prove
\Cref{thm:main_informal}.
We first define a \emph{linearly coupled model} for both \csp{} and \sgxi{}. We connect the free energy of the coupled models via the Guerra-Toninelli interpolation as in several prior works~\cite{panchenko2005free, chen2019suboptimality, jagannath2020unbalanced, chou2022limitations};
our proof is especially inspired by \cite{panchenko2018ksat}.

\begin{definition}[$(A,b)$-coupled models]
\label{defn:coupled-model}
Let $A \in \{0,1\}^{\ell \times \ell'}$ be a $0$-$1$ matrix and $b \in \R_{\geq 0}^{\ell'}$ a nonnegative vector satisfying $(Ab)_i = 1$ for all $i \in [\ell]$.

An \emph{$(A, b)$-coupled model of \sgxi{}} is a collection of random Hamiltonians $\cG_1, \dots, \cG_\el$ related by
\begin{align}
    \left(\begin{array}{c} \cG_1 \\ \vdots \\ \cG_{\ell}
\end{array}\right) 
= A\left(\begin{array}{c} \cG'_1 \\ \vdots \\ \cG'_{\ell'}
\end{array}\right)\,,
\end{align}
where each $\cG'_{i'}$ is the Hamiltonian for an independent instance of $\sg{}_{s \to \sqrt{b_{i'}}\xi(s)}$.
The \emph{grand Hamiltonian} is $\cG: (\{\pm 1\}^n)^\ell \to \mathbb{R}$ defined by $\cG(\mathbf{x}) = \sum_{i=1}^{\ell}\cG_i(\mathbf{x}_i)$.

An \emph{$(A, b)$-coupled model of\ \csp{}} is a collection of random Hamiltonians $\cH_1, \dots, \cH_\el$ related by
\begin{align}
\left(\begin{array}{c} \cH_1 \\ \vdots \\ \cH_{\ell}
\end{array}\right) 
= A\left(\begin{array}{c} \cH'_1 \\ \vdots \\ \cH'_{\ell'}
\end{array}\right)\,,
\end{align}
where each $\cH'_{i'}$ is the Hamiltonian for an independent instance of~\cspdens{b_{i'}\alpha}.
The \emph{grand Hamiltonian} is $\cH: (\{\pm 1\}^n)^\ell \to \mathbb{R}$ defined by $\cH(\mathbf{x}) = \sum_{i=1}^{\ell}\cH_i(\mathbf{x}_i)$.
\end{definition}
In this section we use $\cG$ for a grand Hamiltonian of a spin \underline{g}lass and $\cH$ for that of a CSP.
Expectations are over the randomness of the Hamiltonians.

Notice that $A$ is $\{0,1\}$-valued. One can think of this matrix as choosing which of the $\el'$ independent ``hidden'' instances are connected to each of the $\el$ ``observed'' instances.
The vector $b$ gives the variances of the hidden instances,
and the scaling is set so that the variance of each observed instance is the same as a single instance.

\begin{remark}
    This model describes any linear correlation structure on $\el$ Hamiltonians
    $\cH_1, \dots, \cH_\el$,
    in the sense that it is equivalent to the following alternative model. For the CSP model, there are $2^\el$ parameters $(\eta_S)_{S \subseteq [\el]}$
    such that $\eta_S$ specifies the fraction of constraints that are simultaneously shared between
    all of the Hamiltonians in $S$ and the remaining constraints are sampled independently.
\end{remark}

We study the free energy of our models, restricted to sets of assignments with particular overlap vectors.
\begin{definition}[Free energy of states with given overlap]
    Let $S \subseteq \mathcal{R}^{(\ell)}$ be an open set.
    Let $\cH :(\{\pm 1\}^n)^\el \to \R$ be a grand Hamiltonian. 
    Define $Z_{\cH, S}(\beta)$ as the partition function of $\cH$ when the configuration tuples are restricted to those having overlap vector in $S$; that is,
    \begin{align}
        Z_{\cH, S}(\beta) &= \sum_{\mathbf{x} \in U^{(\ell)}_n(S)} e^{\beta\,\cH(\mathbf{x})}\,.
    \end{align}
Define the free energy density of $\cH$ as
    \begin{align}
        \phi_{\cH, S}(\beta) &= \frac{1}{\beta \ell n} \log Z_{\cH,S}(\beta)\,.
    \end{align}
\end{definition}
\begin{remark}
    We ignore the edge case $U_n^{(\el)}(S) = \emptyset$.
\end{remark}

Now we can state our main interpolation.
\begin{theorem}[Interpolation of random Max-CSPs and spin glasses; generalized version of \Cref{thm:main_informal}]
\label{thm:overlap-geometry}
Choose a positive integer $\ell$.
Consider any set $S \subseteq \cR^{(\ell)}$. Let $\cG$ and $\cH$ be grand Hamiltonians for $(A,b)$-coupled models of \sgxi{} and \csp{}, respectively, where  $\xi$ is related to $\Lam$ as in \Cref{eq:main-eqn}. Then:
\begin{align}
\label{eq:free-energy-equiv}
\E \phi_{\cH,S}(\beta)
=
\E_{f \sim \Lam}[f]
+ 
\frac{1}{\sqrt{\alpha}}
\E \phi_{\cG,S}(\beta)
+ O\left( \tfrac{\beta^2}{\alpha^2}\right) + 
o_n\left(1\right)\,. \end{align}
where the second-to-last term (which may depend on $n$) satisfies $\abs{O\left(\tfrac{\beta^2}{\alpha^{2}}\right)} \leq C \cdot \tfrac{\beta^2}{\al^2}$ whenever $\tfrac{\beta}{\al} \leq \eps_0$ for constants $C, \eps_0 >0$ depending on  $\el$.

Furthermore, by the concentration arguments in \Cref{sec:concentration},
we may drop the expectations over $\cG, \cH$
with the qualitative change that the last term is random and is $o_n\left(1\right)$ with high probability.
\end{theorem}

In order to use the Guerra-Toninelli interpolation, we modify the CSP model so that the
number of constraints is $m \sim \Pois(\al n)$ rather than $m = \al n$ fixed.
Because $\Pois(\al n)$ concentrates around its mean as $n \to \infty$, the free energy density is asymptotically the same.
The (straightforward) proof of the following lemma is in \Cref{app:csp-technical}.
Using this lemma contributes the $o_n(1)$ term to \Cref{eq:free-energy-equiv}.
\begin{restatable}{lemma}{poissonVsExact}
\label{lemma:poisson_equals_fixed_freeenergy}
    Let $\phi^{(\textnormal{Pois})}_{\cH,S}$ and $\phi^{(\textnormal{exact})}_{\cH,S}$ denote the 
    free energy density of a $\csp{}$ instance with $m~\sim~\Pois(\al n)$ and $m~=~\al n$ clauses, respectively. Then
    \begin{align}
        \abs{\E\phi^{(\textnormal{Pois})}_{\cH,S} - \E\phi^{(\textnormal{exact})}_{\cH,S}} \leq \frac{1}{\sqrt{\al n}} = 
        o_n\left(1\right)\,.
    \end{align}
\end{restatable}

\subsection{Proof of \texorpdfstring{\Cref{thm:overlap-geometry}}{Theorem 3.4}}
For notational convenience, we assume in this proof that $\Lam$ is supported
on a single predicate $f$. The general case is handled by converting $\widehat{f}(\emptyset)$ back to $\E_{f \sim \Lam}[f]$.

Fix positive integers $\ell, \ell'$. We define an interpolated Hamiltonian
\begin{align}\label{eq:gen-gt-interpolation}
    \cK(t,\mathbf{x}) = \cH^{\alpha (1-t) n}(\mathbf{x}) - (1-t) \ell n \widehat{f}(\varnothing) + \sqrt{\frac{t}{\alpha}} \cG(\mathbf{x})\,,
\end{align}
where $\cH^{\alpha (1-t) n}$ is the grand Hamiltonian of an $(A,b)$-coupled model of \cspdens{\alpha(1-t)}. The parameter $t$ controls the interpolation from the Max-CSP (when $t = 0$) to the spin glass (when $t=1$).

We let $\beta > 0$ be a parameter independent of $n$.
We will later choose $\beta$ such that $1 \ll \beta \ll \alpha$ as $\alpha \to \infty$.

Fix an open subset $S \subseteq \cR^{(\ell)}$. We write the average free energy density of $\cK$, at inverse temperature $\beta$, among states that produce overlap vectors in $S$:
\begin{align}
    \phi_{\beta}(t) :=& \E_{\cH, \cG}\phi_{\cK(t, \cdot), S}(\beta)\\
    =& \frac{1}{\ell n} \frac{1}{\beta}
    \E_{\cH, \cG} 
    \log 
    \sum_{\mathbf{x} \in U_n(S)} 
    \exp\left( \beta \cK(t,\mathbf{x}) \right)\,.
\end{align}

In this proof, we upper-bound the derivative $\dv{}{t}\phi_{\beta}(t)$, and thereby the difference between $\phi_{\beta}(0)$ and $\phi_{\beta}(1)$, showing that the free energy densities of $\cG$ and $\cH$ are close.

\subsubsection{Taking the derivative of $\phi_{\beta}(t)$}

Let's calculate $\dv{}{t}\phi_{\beta}(t)$.
First, we generalize $\cK$ and $\phi_\beta$ so the $t$-dependence of $\cG$ and $\cH$ are controlled by independent parameters:
\begin{align}
    \widetilde \cK(\rho, \gamma, \mathbf{x}) 
    &:= \cH^{\rho}(\mathbf{x}) - \frac{\rho\ell }{\alpha} \widehat{f}(\varnothing) + \frac{\gamma}{\sqrt{\alpha}} \cG(\mathbf{x})\,,
    \\
    \widetilde\phi_{\beta}(\rho,\gamma) &:= \frac{1}{\beta \el n} \E_{\cH^{\rho},\cG} \log \sum_{\mathbf{x} \in U_n(S)} \exp\left( \beta  \widetilde \cK(\rho,\gamma,\mathbf{x})\right)\,.
\end{align}
When $\rho = (1-t) \alpha n$ and $\gamma = \sqrt{t}$, we recover the original expressions:
\begin{align}
 \cK(t, \mathbf{x}) &= \widetilde \cK((1-t)\alpha n,\sqrt{t},\mathbf{x})
 \\
\phi_{\beta}(t) &= \widetilde \phi_{\beta}((1-t)\alpha n,\sqrt{t})
\end{align}
We then take a derivative using the chain rule:
\begin{align}
\label{eq:coupled-chain-rule}
    \dv{}{t}\phi_{\beta}(t) &=
    \left(\pdv{}{\rho}\widetilde\phi_{\beta}(\rho,\gamma)\right)\dv{\rho}{t} + \left(\pdv{}{\gamma}\widetilde\phi_{\beta}(\rho,\gamma)\right)\dv{\gamma}{t}
    \\
    &= - \alpha n \left(\pdv{}{\rho}\widetilde\phi_{\beta}(\rho,\gamma)\right) +
    \frac{1}{2 \sqrt{t}}\left(\pdv{}{\gamma}\widetilde\phi_{\beta}(\rho,\gamma)\right)
\end{align}

We also introduce a Gibbs expectation operator to use when computing the partial derivatives. Given a function $p(\mathbf{x})$, the average of $p$
with respect to the Gibbs distribution of $\widetilde{\cK}(\rho, \gamma, x)$ is
\begin{align}
     \iprod{p}_{\mathbf{x}} := \frac{
     \sum_{\mathbf{x} \in U_n(S)}
     p(\mathbf{x})
     \cdot \exp\left(\beta \widetilde \cK(\rho, \gamma,\mathbf{x})
     \right)
     }
     {\sum_{\mathbf{x} \in U_n(S)} \exp\left(\beta \widetilde \cK(\rho, \gamma,\mathbf{x})\right)}\,.
\end{align}

\subsubsection{The Poisson derivative}
To calculate 
$\pdv{}{\rho}\widetilde\phi_{\beta}(\rho,\gamma)$, we introduce $\ell'$ new variables $\rho_1, \dots, \rho_{\ell'}$ to parameterize the independent Poisson instances used to construct $\cH^{\rho}$. 
We introduce intermediate functions
\begin{align}
    \widetilde \cK(\rho_1, \dots,\rho_{\ell'},\gamma,\mathbf{x})   &=  
    \sum_{i\in[\ell]} \bigg[ \bigg( A\left(\begin{array}{c} \cH'^{\rho_1}_1 \\ \vdots \\ \cH_{\ell'}'^{\rho_{\ell'}}
\end{array}\right)\bigg)_i(\mathbf{x}_i)
    - \bigg(A\left(\begin{array}{c} \rho_1 \\ \vdots \\ \rho_{\ell'}
\end{array}\right)\bigg)_i{\frac{\widehat{f}(\varnothing)}{\alpha}}+
 \frac{\gamma}{\sqrt{\alpha}} \cG_i(\mathbf{x}_i)
    \big]\,,
  \\
   \widetilde\phi_{\beta}(\rho_1, \dots, \rho_{\ell'},\gamma) &= \frac{1}{\beta \el n} \E_{\cH_1'^{\rho_1},\dots,\cH_{\ell'}'^{\rho_{\ell'}}}\E_{\cG}
   \bigg[ 
   \log \sum_{\mathbf{x} \in U^{(\ell)}_n(S)} \exp\left( \beta \widetilde \cK(\rho_1, \dots, \rho_{\ell'},\gamma,\mathbf{x})\right)
   \bigg]\,.
  \end{align}
When $\rho_{i'} = b_{i'}\rho$ for all $i'\in[\ell']$, we recover the original functions $\widetilde \cK$ and $\widetilde \phi_\beta$. In this case, through another application of the chain rule,
\begin{align}
    \pdv{}{\rho}\widetilde\phi_{\beta}(\rho,\gamma)
= \sum_{i' \in [\ell']}
\bigg(\pdv{}{\rho_{i'}}\widetilde\phi_{\beta}(\rho_1, \dots, \rho_{\ell'},\gamma)\bigg)\pdv{\rho_{i'}}{\rho}
= \sum_{i \in [\ell']} b_{i'}\pdv{}{\rho_{i'}}\widetilde\phi_{\beta}(\rho_1, \dots, \rho_{\ell'},\gamma)\,.
\end{align}
We use the explicit derivative of a Poisson variable:
\begin{fact}[Derivative of a Poisson variable]\label{fact:poisson-derivative}
\begin{align}\pdv{}{\lam} \E_{X \sim \Pois(\lam)} \big[f(X)\big] = \E_{X \sim \Pois(\lam)}\big[f(X+1) - f(X)\big]
\end{align}
\end{fact}
Because of \Cref{fact:poisson-derivative},
\begin{align}
&\pdv{}{\rho_{i'}}\widetilde\phi_{\beta}(\rho_1, \dots, \rho_{\ell'},\gamma)  
= 
\widetilde\phi_{\beta}(\rho_1, \dots, \rho_{i'}+1, \dots, \rho_{\ell'},\gamma) - \widetilde\phi_{\beta}(\rho_1, \dots, \rho_{i'}, \dots, \rho_{\ell'},\gamma)
\\
&=  \frac{1}{\beta \ell n} \E_{\cH_1'^{\rho_1},\dots,\cH_{\ell'}'^{\rho_{\ell'}}}\E_{\cG}\log
\bigg\langle
\exp\bigg(
\beta
\widetilde \cK(\rho_1, \dots, \rho_{i'}+1,\dots,\rho_{\ell'},\gamma,\mathbf{x}) 
-
\beta
\widetilde \cK(\rho_1, \dots, \rho_{i'},\dots,\rho_{\ell'},\gamma,\mathbf{x})
\bigg)
\bigg\rangle_{\mathbf{x}}
\\
&= \frac{1}{\beta \ell n}
\E_{\cH_1'^{\rho_1},\dots,\cH_{\ell'}'^{\rho_{\ell'}}}
\log\bigg\langle
\exp
\bigg(
\beta
\sum_{j\in[\ell]} A_{ji'} 
\big(
(\cH_{i'}'^{\rho_{i'}+1} - \cH_{i'}'^{\rho_{i'}})(\mathbf{x}_j) - \widehat{f}(\varnothing)
\big)
\bigg)
\bigg\rangle_{\mathbf{x}}\, .
\end{align}
The difference of $\cH_{i'}'^{\rho_{i'}+1}$ and $\cH_{i'}'^{\rho_{i'}}$ is a single extra clause (normalized by $\alpha$), applied to $k$ random indices of the input with random signs.
Let $f^*$ be the extra clause. Then
\begin{align}
    =  \frac{1}{\beta  \ell n}
    \E_{f^*}
\log\iprod{
\exp\left(
\frac{\beta}{\alpha}
\sum_{j\in[\ell]} A_{ji'} 
\big(
f^*(\mathbf{x}_j) - \widehat{f}(\varnothing)
\big)
\right)
}_{\mathbf{x}}\,.
\end{align}

When $\frac{\beta}{\alpha}$ is small, this term can be Taylor-expanded as  $\log z = \log\left(1 - (1 - z)\right)= - \sum_{p \ge 1} \frac{(1-z)^p}{p}$:
\begin{align}
 &=  \frac{-1}{\beta\ell n}
\E_{f^*}
\sum_{p \ge 1} \frac{1}{p}\bigg( 1 - 
\iprod{
\exp\left(
\frac{\beta}{\alpha}
\sum_{j\in[\ell]}
A_{ji'}
\big(
f^*(\mathbf{x}_j) - { \widehat{f}(\varnothing)}
\big)
\right)
}_{\mathbf{x}}
\bigg)^p\,.
\end{align}
We introduce additional ``replicas'' $\mathbf{x}_{(s)}$ of $\mathbf{x}$. Precisely, each $\mathbf{x}_{(s)}$ is an i.i.d. copy of $\mathbf{x}$.
For any function $w$ and any set of replicas $y_0,y_1,\dots$, we have the identity $\iprod{w(y_0)}^p = \iprod{\prod_{s \in [p]}w(y_s)}$. As a result,
\begin{align}
 &=  \frac{-1}{\beta \ell n}
\E_{f^*} 
\iprod{
\sum_{p \ge 1}
\frac{1}{p}
\prod_{s=1}^p 
\bigg( 1 -
\exp\left(
\frac{\beta}{\alpha}
\sum_{j\in[\ell]}
A_{ji'}
\big(
f^*(\mathbf{x}_{(s)_j}) - { \widehat{f}(\varnothing)}
\big)
\right)
\bigg)
}_{\mathbf{x}_{(1)}, \mathbf{x}_{(2)}, \dots}\,.
\end{align}
We Taylor-expand this expression in $\frac{\beta}{\alpha}$. The first line comes from the $p=1$ term and the second line comes from the $p=2$ term:
\begin{align}
 =  \frac{-1}{\beta \ell n}
\E_{f^*}
\bigg\langle
&- \frac{\beta}{\alpha} \sum_{j\in[\ell]} A_{ji'} 
\big(
f^*(\mathbf{x}_j) - { \widehat{f}(\varnothing)}
\big) 
-
\frac{\beta^2}{2\alpha^2} 
\bigg(
\sum_{j\in[\ell]} A_{ji'} 
\big(f^*(\mathbf{x}_j) - { \widehat{f}(\varnothing)}
\big)
\bigg)^2 \nonumber
\\
&+ 
\frac{\beta^2}{2\alpha^2} 
\bigg(
\sum_{j\in[\ell]} A_{ji'} 
\big(
f^*(\mathbf{x}_j) - { \widehat{f}(\varnothing)}
\big) 
\bigg)
\bigg(
\sum_{j\in[\ell]} A_{ji'} 
\big(
f^*(\mathbf{y}_j) - { \widehat{f}(\varnothing)}
\big) 
\bigg)
+ O\left( \tfrac{\beta^3}{\alpha^3}\right)
\bigg\rangle_{\mathbf{x}, \mathbf{y}}\,.
\end{align}
Notice that because the clause applies random signs to its input,
\begin{align}
\label{eqn:whyweneedrandomsigns}
    \E_{f^*} \big[ \big \langle f^*(\mathbf{x}_j) \big \rangle_{\mathbf{x}_j} \big] = \widehat{f}(\varnothing)\,.
\end{align}
As a result, the only terms that remain are 
\begin{align}
 =  \frac{- \beta }{2 \alpha^2 \ell n} 
\E_{f^*}
\bigg\langle
&
\sum_{i,j\in[\ell]} A_{ii'}A_{ji'} 
\big(
-
f^*(\mathbf{x}_i) f^*(\mathbf{x}_j) 
+ 
f^*(\mathbf{x}_i) f^*(\mathbf{y}_j) 
\big)
\bigg\rangle_{\mathbf{x}, \mathbf{y}}
+  \frac{1}{ \ell n} O\left( \tfrac{\beta^2}{\alpha^3}\right)\,.
\end{align}  

From here, we rewrite the correlation of $f^*$ as a function of the noise stability of $f$. Specifically, let $u_a = \eps_a(\mathbf{x}_i)_{d_a}$ and $v_a = \eps_a(\mathbf{y}_j)_{d_a}$ for some uniformly chosen $\eps \unif \{\pm 1\}^k$ and $d \unif [\ell]^k$. Then $u,v$ are marginally uniform points in the hypercube that are $R(\mathbf{x}_{i},\mathbf{y}_{j})$-correlated. 

Write $f^*(\sigma) = f(\eps_1\sigma_{d_1}, \dots, \eps_k\sigma_{d_k})$. Then 
\begin{align}
    \iprod{f^*(\mathbf{x}_i)f^*(\mathbf{y}_j)}_{\mathbf{x}_i,\mathbf{y}_j}  = \iprod{f(u)f(v)}_{\mathbf{x}_i,\mathbf{y}_j} = 
    \big\langle \Stab_{R(\mathbf{x}_i,\mathbf{y}_j)}[f] \big\rangle_{\mathbf{x}_i,\mathbf{y}_j}
\end{align}
where $\Stab_{R(\mathbf{x}_i,\mathbf{y}_j)}[f]$ is defined as in~\Cref{defn:noise-stability}. Using this, we get
\begin{align}
&= \frac{- \beta}{2 \alpha^2 \ell n} 
\sum_{i,j\in[\ell]} A_{ii'}A_{ji'} 
\big\langle
- \Stab_{R(\mathbf{x}_{i}, \mathbf{x}_{j})}[f]
+ \Stab_{R(\mathbf{x}_{i}, \mathbf{y}_{j})}[f]
\big\rangle_{\mathbf{x}, \mathbf{y}}
+  \frac{1}{ \ell n} O\left( \tfrac{\beta^2}{\alpha^3}\right)
\\
&= \frac{\beta}{2 \alpha^2 \ell n} 
\sum_{i,j\in[\ell]} A_{ii'}A_{ji'} 
\big\langle
 \xi(R(\mathbf{x}_i, \mathbf{x}_j))
- \xi(R(\mathbf{x}_i, \mathbf{y}_j))
\big\rangle_{\mathbf{x},\mathbf{y}}
+  \frac{1}{ \ell n} O\left( \tfrac{\beta^2}{\alpha^3}\right)\,.
\end{align}

We can now write the $\rho$-derivative of $\widetilde{\phi}_\beta$:
\begin{align}
    \pdv{}{\rho}\widetilde\phi_{\beta} =\frac{\beta}{2 \alpha^2 \ell n}  \sum_{i' \in [\ell']} \sum_{i,j \in [\ell]} b_{i'} A_{ii'}A_{ji'} 
\big\langle
 \xi(R(\mathbf{x}_i, \mathbf{x}_j))
- \xi(R(\mathbf{x}_i, \mathbf{y}_j))
\big\rangle_{\mathbf{x},\mathbf{y}} 
+  \frac{1}{ \ell n} O\left( \tfrac{\beta^2}{\alpha^3}\right)
\end{align}
\subsubsection{The Gaussian derivative}

We compute the derivative of $\widetilde\phi_{\beta}(\rho, \gamma)$ with respect to $\gam$. 
Since the variance of each $\cG_i$ is independent of $\gamma$, 
the derivative pulls into the expectation operator: 
\begin{align}
    &\pdv{\gam}\widetilde\phi_{\beta}(\rho, \gamma) = \frac{1}{\beta \ell n}\E_{\cH^\rho, \cG }\pdv{\gam}\log\ \bigg[ \sum_{\mathbf{x} \in U_n(S)}\exp\left(\beta
    \widetilde{\cK}(\rho, \gam, \mathbf{x})\right) \bigg]
\end{align}
By the chain rule, the $\gamma$-derivative of the partition function is proportional to the Gibbs average of the $\gamma$-derivative of the Hamiltonian $\widetilde{\cK}$:
\begin{align}
    &= \frac{1}{\ell n}
    \E_{\cH^\rho, \cG}
    \left[
    \left\langle 
    \pdv{\gam}
    \left(
    \widetilde{\cK}(\rho, \gam, \mathbf{x})
    \right)
    \right\rangle_{\mathbf{x}}
    \right] 
\end{align}

Notice that the expectation is a correlation between a Gaussian process and a Gibbs measure, so we can use the following formula:
\begin{lemma}[Stein's lemma with Gibbs average~{\cite[Lemma 1.1]{panchenko2013sherrington}}]\label{lem:steins-with-gibbs}
    Consider two jointly Gaussian processes $\{Y(\sigma)\}_\sigma$ and $\{Z(\sigma)\}_\sigma$. For any $w$, let $\langle \cdot \rangle_{\sigma_1,\dots,\sigma_w}$ be the $w$-product Gibbs measure with respect to the process $\{Z(\sigma)\}_\sigma$. Then
    \begin{align}
        \E_{Y,Z}[\langle Y(\sigma) \rangle_\sigma] = \E_{Z}\left[\bigg\langle\E_{Y}\bigg[Y(\sigma)Z(\sigma)\bigg]\bigg\rangle_\sigma - \bigg\langle \E_{Y}\bigg[Y(\sigma)Z(\sigma')\bigg]\bigg\rangle_{\sigma, \sigma'}\right]\, .
    \end{align}
\end{lemma}
Consider the following Gaussian processes:
\begin{align}
    S_{\cH^\rho, \cG}: (\mathbf{x}) &\to \pdv{\gam}
    \left(
    \widetilde{\cK}(\rho, \gam, \mathbf{x})
    \right)
    \\
    T_{\cH^\rho, \cG}: (\mathbf{x}) &\to \beta
    \widetilde{\cK}(\rho, \gam, \mathbf{x})
\end{align}
Applying~\Cref{lem:steins-with-gibbs} then yields
    \begin{align}
        = \frac{1}{\ell n} 
        \E_{\cH'^{\rho}, \cG}
        \bigg[ 
        \left \langle 
        \E_{\cH'^{\rho}, \cG}
        \big[ 
        S_{\cH^\rho, \cG}(\mathbf{x})
        T_{\cH^\rho, \cG}(\mathbf{x})
        \big]
        \right \rangle_{\mathbf{x}}
        - \ 
        \left \langle
        \E_{\cH'^{\rho}, \cG}
        \big[ 
        S_{\cH^\rho, \cG}(\mathbf{x})
        T_{\cH^\rho, \cG}(\mathbf{y})
        \big]
        \right \rangle_{\mathbf{x},\mathbf{y}}
        \bigg]\,.
    \end{align}

The derivative only acts on the Gaussian components of the Hamiltonian $\widetilde{\cK}$:
    \begin{align}
    \qquad\qquad= 
    \frac{1}{\sqrt{\alpha} \ell n}
    \E_{\cH^\rho, \cG }
    \bigg[
    &\bigg\langle
      \E_{\cG}
        \big[ 
        \cG(\mathbf{x})
     T_{\cH^\rho, \cG}(\mathbf{x})
    \big]
    \bigg\rangle_{\mathbf{x}} 
    - 
    \bigg\langle
     \E_{\cG}
        \big[ 
        \cG(\mathbf{x})
     T_{\cH^\rho, \cG}(\mathbf{y})
     \big]
    \bigg\rangle_{\mathbf{x},\mathbf{y}}
    \bigg]
    \\ 
     = 
     \frac{\beta}{\sqrt{\alpha} \ell n}
     \E_{\cH^\rho, \cG }
     \bigg[
     &\left\langle
        \E_{\cG} \big[
        \cG(\mathbf{x})
        \big(
        \cH^{\rho}(\mathbf{x}) -\  \frac{\rho \ell}{\alpha} \widehat{f}(\varnothing)
        + \frac{\gamma}{\sqrt{\alpha}}\cG(\mathbf{x})
        \big)
     \big]
     \right\rangle_{\mathbf{x}} - 
     \nonumber \\  
     &
     \bigg\langle 
     \E_{\cG} \big[
             \cG(\mathbf{x})
        \big(
        \cH^{\rho}(\mathbf{y}) - \frac{\rho \ell}{\alpha} \widehat{f}(\varnothing)
        + \frac{\gamma}{\sqrt{\alpha}}\cG(\mathbf{y})
        \big)
     \big]
     \bigg\rangle_{\mathbf{x},\mathbf{y}}
     \bigg]
    \end{align}
The constant terms (proportional to $\widehat{f}(\varnothing)$) cancel.
Furthermore, $\cG$ is centered and is independent of $\cH$: For each independent instance $\cG'_i$, $\E[\cG'_{i}\cH^\rho] = 0$, and by linearity of expectation, $\E[\cG \cH^\rho] = 0$.
All that remains is
\begin{align}
    &= 
    \frac{\beta \gamma}{\alpha\ell n}
    \E_{\cG }
    \bigg[
    \bigg\langle
     \E_{\cG} 
     \big[
    \cG(\mathbf{x})\cG(\mathbf{x})
     \big]
     \bigg\rangle_{\mathbf{x}}
    - \ \ 
   \bigg\langle
    \E_{\cG} 
     \big[
    \cG(\mathbf{x})\cG(\mathbf{y})
   \big]
   \bigg\rangle_{\mathbf{x},\mathbf{y}}
   \bigg]\,.
    \end{align}
    Using the definition of an $(A,b)$-coupled model, this is
    \begin{align}
    = \frac{\beta\gamma}{ \alpha \ell n}\E_{\cG}\bigg[
    &\bigg\langle\sum_{i,j \in [\ell]}\sum_{i',j' \in [\ell']}A_{ii'}A_{ji'}\sqrt{b_{i'}b_{j'}}
    \E_{\cG}[\cG'_{i'}(\mathbf{x}_i)\cG'_{j'}(\mathbf{x}_j)]
    \bigg\rangle_{\mathbf{x}} - \nonumber \\
    &\bigg\langle\sum_{i,j\in [\ell]}
    \sum_{i',j' \in [\ell']}
    A_{ii'}A_{ji'}\sqrt{b_{i'}b_{j'}}\E_{\cG}[\cG'_{i'}(\mathbf{x}_i)\cG'_{j'}(\mathbf{y}_j)]\bigg\rangle_{\mathbf{x},\mathbf{y}}\bigg]\,.
    \end{align}
   By \Cref{fact:cov-gaussian-process}, we can write the covariance as a function of $\xi$:
    \begin{align}
        \E_{\cG}[\cG'_{i'}(\mathbf{x}_i) \cG'_{j'}(\mathbf{y}_j)] =
        \begin{cases}
n \xi(R(\mathbf{x}_i,\mathbf{y}_j)) & i' = j' \\
0 & \text{otherwise}
\end{cases}
    \end{align}
    Thus, we have
    \begin{align}
        = 
         \frac{\beta\gam}{ \alpha \ell}
    \E_{\cG}
        \bigg[
        \bigg \langle 
        \sum_{i,j \in [\ell]}
        \sum_{i' \in [\ell']}
        A_{ii'}A_{ji'}b_{i'}
              \big(
        \xi(R(\mathbf{x}_i,\mathbf{x}_j)) 
         -
   \xi(R(\mathbf{x}_i, \mathbf{y}_j))
    \big)
    \bigg\rangle_{\mathbf{x},\mathbf{y}}\bigg]\,.
    \end{align}

\subsubsection{Putting it all together}
Now we can calculate the total derivative of $\phi_{\beta}(t)$:
\begin{align}
 \dv{}{t}\phi_{\beta}(t) &=
    -\alpha n \left(\pdv{}{\rho}\widetilde\phi_{\beta}(\rho,\gamma)\right) +
    \frac{1}{2\sqrt{t}}\left(\pdv{}{\gamma}\widetilde\phi_{\beta}(\rho,\gamma)\right)
    \\
    &=\big(\frac{-\beta}{2 \alpha \ell} + \frac{\beta \gamma}{2 \alpha \ell \sqrt{t}}  \big)
    \sum_{i' \in [\ell']} \sum_{i,j \in [\ell]} b_{i'} A_{ii'}A_{ji'} 
\big\langle
 \xi(R(\mathbf{x}_i, \mathbf{x}_j))
- \xi(R(\mathbf{x}_i, \mathbf{y}_j))
\big\rangle_{\mathbf{x},\mathbf{y}}   - \frac{1}{\ell} O\left( \tfrac{\beta^2}{\alpha^2}\right)
\\
&= O\left( \tfrac{\beta^2}{\alpha^2}\right)\,,
\end{align}
since $\gamma = \sqrt{t}$ and $\ell$ is a constant.
So 
\begin{align}
    \abs{\phi_\beta(1) - \phi_\beta(0)} 
    \le \max_{t \in [0,1]} \dv{}{t} \phi_\beta(t) =  O\left( \tfrac{\beta^2}{\alpha^2}\right) \,.
\end{align}
This proves \Cref{thm:overlap-geometry}.

\section{Optimal value of a random Max-CSP}
\label{sec:optimalvalue}

As a corollary of \Cref{thm:overlap-geometry}, we prove that
in the large $\alpha$ limit,
the optimal value of a coupled Max-CSP among solutions with given overlap structure is determined by that of a spin glass.
\begin{corollary}
\label{prop:max-value}
    The following holds for all $\el, \el'$ and $(A,b)$-coupled models $\cG$ of \sgxi{} and $\cH$ of $\csp{}$,
    and all open sets $S \subseteq \cR^{(\el)}$,
    \begin{align}
        \E\max_{\mathbf{x} \in U^{(\el)}_n(S)} \frac{1}{\el n} \cH(\mathbf{x}) = \E_{f \sim \Lam}[f] +  \E\max_{\mathbf{x} \in U^{(\el)}_n(S)}\frac{1}{\sqrt{\al}} \frac{1}{\el n} \cG(\mathbf{x}) + O(\alpha^{-2/3}) + o_n(1)\,.
    \end{align}
    where the second-to-last term (which may depend on $n$) satisfies $\abs{O\left(\alpha^{-2/3}\right)} \leq C \cdot \alpha^{-2/3}$ for all $\alpha \geq C$ for a constant $C$ depending on $\el$.

    Furthermore, by the concentration arguments in \Cref{sec:concentration},
    we may drop the expectations over $\cG, \cH$
    with the qualitative change that the last term is random and is $o_n\left(1\right)$ with high probability.
\end{corollary}

\begin{proof}
Using \Cref{thm:overlap-geometry}, 
we have
\begin{align}\E \phi_{\cH,S}(\beta) = \E_{f \sim \Lam}[f] + \frac{1}{\sqrt{\alpha}} \E \phi_{\cG,S}(\beta) + O\left(\tfrac{\beta^2}{\alpha^2}\right) +o_n(1) \,.\end{align}
Recall that there are at most $2^{\ell n}$ choices of $\mathbf{x}$. Applying \Cref{fact:maxvalue} for $\cG$ and for $\cH$, we get
\begin{align}
    \abs{\phi_{\cG, S}(\beta) - \max_{\mathbf{x} \in U_n(S)} \frac{1}{\el n}\cG(\mathbf{x})} & \leq \frac{\log |U_n(S)|}{\beta \el n} = O\left(\tfrac{1}{\beta}\right) \,,\\
    \abs{\phi_{\cH, S}(\beta) - \max_{\mathbf{x} \in U_n(S)} \frac{1}{\el n}\cH(\mathbf{x})} &= O\left(\tfrac{1}{\beta}\right) \,.
\end{align}
By Jensen's inequality, the same holds when using expectations.
Plugging in these bounds and choosing $\beta = \alpha^{2/3}$ proves the claim.
\end{proof}
In the special case $\ell = \ell' = 1$, $(A, b) = (\mathbb{I}, (1))$, and $S = \cR^{(\ell)} = [-1,1]$,
we conclude \Cref{cor:maximum_value}.

\subsection{Numerical calculations}
We compute the value of the Parisi formula for the spin glasses associated with popular CSPs in \Cref{table:calcs}. This code can be run for any choice of spin glass and is available online\footnote{ \href{https://github.com/marwahaha/csp-parisi/}{https://github.com/marwahaha/csp-parisi/}}.
Our code uses the zero-temperature representation of the Parisi functional from \cite{auffinger2017parisi}, which we restate below:

\begin{definition}[Parisi functional at zero temperature~\cite{parisi1980sequence, auffinger2017parisi}]
\label{defn:parisivariational}
        Given a function $\zeta$ in 
\begin{align}
        \cU = \left\{\zeta : [0,1) \to \R_{\geq 0} : \zeta \text{ is right-continuous, non-decreasing, }\int_{0}^1 \zeta(t)dt < \infty \right\}\,,
    \end{align}
    the Parisi functional $\cP^\xi_{\infty}(\zeta)$ is
    \begin{align}
        \cP^\xi_{\infty}(\zeta) = \Phi_\zeta(0,0) - \frac{1}{2}\int_0^1s\xi''(s)\zeta(s)ds\, ,
    \end{align}
    where the function $\Phi_\zeta(x;t): \R \times [0,1] \to \R$ is the solution of the Hamilton-Jacobi-Bellman equation
    \begin{align}
        \partial_t\Phi_\zeta(x;t) + \frac{\xi''(t)}{2}\left(\partial_{xx}\Phi_\zeta(x;t) + \zeta(t)\left(\partial_x\Phi_\zeta(x;t)\right)^2\right) = 0\, ,
    \end{align}
    with initial condition $\Phi_\zeta(x;1) = \abs{x}$.
\end{definition}

As stated in the introduction, the optimum value of a spin glass is
$\inf_{\zeta \in \cU} \cP^{\xi}_\infty(\zeta)$.
The value ALG achievable by efficient algorithms is given by $\inf_{\zeta \in \cL} \cP^{\xi}_\infty(\zeta)$ for an extended class of functions $\cL \supseteq \cU$ defined below. See \cite[Section 6.1]{huang2021tight} and references therein for well-posedness considerations.

\begin{definition}[ALG]\label{def:alg}
    Let ALG $= \inf_{\zeta \in \cL} \cP^{\xi}_\infty(\zeta)$,
    where the class of functions $\cL$ is defined by
    \begin{equation}
        \cL = \left\{\begin{array}{cc}
             \zeta : [0, 1) \to \R_{\geq 0}: \zeta\text{ is right-continuous, }\\
             \norm{\xi" \cdot \zeta}_{TV[0,t]} < \infty \text{ for all }t \in [0,1),\int_0^1 \xi"(t)\zeta(t) dt < \infty
        \end{array}\right\}
    \end{equation}
    where the total variation $\norm{f}_{TV(J)}$ on an interval $J$ is
    \begin{equation}
        \norm{f}_{TV(J)} = \sup_{n \in \N} \sup_{t_0 < \cdots < t_n, t_i \in J} \sum_{i =1}^n \abs{f(t_i) - f(t_{i-1})}\,.
    \end{equation}
\end{definition}

\begin{theorem}[{\cite[Lemma 5.4]{chen2019suboptimality}}]
\label{thm:alg-opt}
    For even mixture polynomials $\xi$ such that $c_2 = 0$, $\textsf{ALG} < \textsf{GSED}(\sgxi{})$ holds with strict inequality.
\end{theorem}

\begin{table}[ht]
\centering
\renewcommand{\arraystretch}{1.1}
\begin{tabular}{|c|c|c|c|c|}
\hline
\textbf{$k$} & \textbf{Max 1-in-$k$ SAT}                   & \textbf{Max $k$ NAESAT}                      & \textbf{Max $k$ SAT}                         & \textbf{Max $k$ XOR}                       \\ \hline
2            & $\frac{1}{2} + \frac{0.54}{\sqrt{\alpha}}$  & $\frac{1}{2} + \frac{0.54}{\sqrt{\alpha}}$   & $\frac{3}{4} + \frac{0.40}{\sqrt{\alpha}}$   & $\frac{1}{2} + \frac{0.54}{\sqrt{\alpha}}$ \\ \hline
3            & $\frac{3}{8} + \frac{0.54}{\sqrt{\alpha}}$  & $\frac{3}{4} + \frac{0.47}{\sqrt{\alpha}}$   & $\frac{7}{8} + \frac{0.33}{\sqrt{\alpha}}$   & $\frac{1}{2} + \frac{0.58}{\sqrt{\alpha}}$ \\ \hline
4            & $\frac{1}{4} + \frac{0.48}{\sqrt{\alpha}}$  & $\frac{7}{8} + \frac{0.37}{\sqrt{\alpha}}$   & $\frac{15}{16} + \frac{0.26}{\sqrt{\alpha}}$ & $\frac{1}{2} + \frac{0.58}{\sqrt{\alpha}}$ \\ \hline
5            & $\frac{5}{32} + \frac{0.41}{\sqrt{\alpha}}$ & $\frac{15}{16} + \frac{0.28}{\sqrt{\alpha}}$ & $\frac{31}{32} + \frac{0.20}{\sqrt{\alpha}}$ & $\frac{1}{2} + \frac{0.59}{\sqrt{\alpha}}$ \\ \hline
\end{tabular}
    \caption{\footnotesize Optimal value of a random $k$-CSP with $n$ variables and $\alpha n$ clauses, as $n \to \infty$. 
    The calculation uses the zero-temperature Parisi functional in \Cref{defn:parisivariational}. 
    All values are rounded to two decimal places.
    The values for $k$XOR match those in \cite{marwaha_kxor}.
    We expect the values to be accurate to two significant figures, based on consistency with independently calculated values for $2$XOR and $3$XOR~\cite{alaoui2020optimization}.
    }
    \label{table:calcs}
\end{table}

\section{Overlap gaps in a random Max-CSP}
\label{sec:ogp_from_freeenergy}

\Cref{thm:overlap-geometry} implies that many quantities are equivalent for \csp{} at large $\alpha$ and its associated spin glass model. For example:
\begin{enumerate}
    \item Free energy of a single instance (proving \Cref{thm:main_informal}).
    \\($\ell = 1$, $\ell' = 1$, $A = [[1]]$, $b = [1]$, $S = \cR$)
    \item Free energy of a single instance, restricting the Hamming weight to $W \subseteq [-1,1]$ (related to \cite{jagannath2020unbalanced}).
    \\($\ell = 1$, $\ell' = 1$, $A = [[1]]$, $b = [1]$, $S = \{r \in \cR: r_{\{1\}} \in W\}$)
    \item Existence (or non-existence) of an overlap gap property in the overlap range $(s,t)$.
    \\    ($\ell = 2$, $\ell' = 1$, $A = [[1],[1]]$, $b = [1]$, $S = \{r \in \cR: s < r_{\{1,2\}} < t\}$)
    \item Existence (or non-existence) of an overlap gap property in the overlap range $(s,t)$ for $\eta$-coupled Hamiltonians~\cite{chen2019suboptimality}.
    \\    ($\ell = 2$, $\ell' = 3$, $A = [[1,\;1,\;0],[1,\;0,\;1]]$, $b = [\eta,\;1-\eta,\;1-\eta]$, $S = \{r \in \cR: s < r_{\{1,2\}} < t\}$)
    \item Existence (or non-existence) of the branching overlap gap property in~\cite{huang2021tight}.
    \\ (Choice of parameters in \Cref{remark:branching_ogp})
\end{enumerate}

We use \Cref{thm:overlap-geometry} to show how \csp{} inherits the overlap gap property from a spin glass.
This proof is a generalization of~\cite[Proof of Theorem 5]{chen2019suboptimality} and~\cite[Lemma 8.14]{chou2022limitations} to hold for arbitrary $(A,b)$-coupled instances. 
We use the same argument to transfer the branching OGP of \cite{huang2021tight}, a hierarchical style of OGP described by a rooted tree.

\begin{definition}[OGP~\cite{gamarnik2021topological}]\label{defn:ogp-with-value}
    A Hamiltonian $H: \{\pm 1\}^n \to \R$ exhibits the overlap gap property (OGP)
    at value $v$ if
    there are $-1 \leq s < t \leq 1$ such that for all $\sigma_1, \sigma_2$ with $H(\sigma_1) \geq v, H(\sigma_2) \geq v$, we have 
    \begin{align}
        R(\sigma_1, \sigma_2) \not\in (s,t)\,.
    \end{align}
\end{definition}

\begin{definition}[Average-OGP, nonstandard definition]\label{defn:average-ogp-with-value}
A Hamiltonian $H$ exhibits the average-OGP at value $v$ if the same holds whenever $\frac{1}{2}(H(\sigma_1) + H(\sigma_2)) \geq v$.
\end{definition}

\begin{remark}
    The average-OGP implies the OGP. On the other hand,
    the OGP implies the average-OGP with a weakened (larger) value.
    The interpolation in this work transfers the average-OGP.
\end{remark}

\begin{proposition}\label{prop:inherit-overlap-gaps}
    If \sgxi{} exhibits the average-OGP at value $v$ with high probability, then for all $\eps > 0$, for all
    sufficiently large $\alpha$, $\csp{}$ exhibits the average-OGP at value $\E_{f \sim \Lam}[f] + \frac{v+\eps}{\sqrt{\alpha}}$ with high probability, when $\xi$ is related to $\Lam$ as in \Cref{eq:main-eqn}.
\end{proposition}
\begin{proof}
    Let $\el =2$, and let $S = (a,b)$ be the overlap gap for the spin glass model \sgxi{}. 
    Using \Cref{prop:max-value}, for sufficiently large $\alpha$, we have w.h.p
    \begin{equation}
    \max_{\mathbf{x} \in U_n^{(\el)}(S)}\frac{1}{\el n}\cH(\mathbf{x}) \leq \E_{f \sim \Lam}[f] + \frac{\max_{\mathbf{x} \in U_n^{(\el)}(S)}\frac{1}{\el n}\cG(\mathbf{x}) + \eps}{\sqrt{\alpha}} +o_n(1)\,.
    \end{equation}
    The overlap gap property
    implies that w.h.p
    \begin{equation}
        \max_{\mathbf{x} \in U_n^{(\el)}(S)}\frac{1}{\el n}\cG(\mathbf{x}) \leq v\,.
    \end{equation}
    Hence 
    with high probability we also have
    \begin{align}
        \max_{\mathbf{x} \in U_n^{(\el)}(S)}\frac{1}{\el n}\cH(\mathbf{x}) \leq \E_{f \sim \Lam}[f] + \frac{v + 2\eps}{\sqrt{\alpha}}\,.
    \end{align}
\end{proof}

Analogously, we can transfer a generic version of the OGP on $(A, b)$-coupled models.
\begin{definition}[Generic OGP]
    For a relatively open subset $S \subseteq \cR^{(\el)}$, random Hamiltonians $H_1, \dots, H_\el: \{\pm 1\}^n \to \R$ exhibit an \emph{$S$-OGP} at value $v$
    if 
    \begin{equation}
        \max_{(\sigma_1, \dots, \sigma_\el) \in U_n^{(\el)}(S)} \frac{1}{\el n}\left(H_1(\sigma_1) + \cdots + H_\el(\sigma_\el)\right) \leq v\,.
    \end{equation}
    The \emph{size} of the OGP is $\el$.
\end{definition}

\begin{proposition}\label{prop:S-OGP}
    If an $(A,b)$-coupled model of \sgxi{} exhibits an $S$-OGP at value $v$ with high probability,
    then for all $\eps > 0$,
    for all sufficiently large $\alpha$,
    the $(A,b)$-coupled model of \csp{} exhibits an $S$-OGP at value $\E_{f \sim \Lam}[f] + \frac{v+\eps}{\sqrt{\al}}$ with high probability,
    when $\xi$ is related to $\Lam$ as in \Cref{eq:main-eqn}.
\end{proposition}

\subsection{The branching OGP}
\label{remark:branching_ogp}

The branching OGP gives bounds \emph{tight} algorithmic bounds for a class of spin glass models, matching the performance of certain approximate message passing algorithms \cite{huang2021tight}. We present the somewhat involved definition and show that it is captured by our framework.

\begin{definition}[Tree-coupled ensemble]
A \emph{tree-coupled ensemble of Hamiltonians} is defined by:
\begin{enumerate}
    \item A rooted tree of height $D$ defined by the vector $\vec{k} \in \Z_+^D$,
    so that every node at depth $d$ has $k_{d+1}$ children.
    We describe the nodes of the tree as $\mathbb{T}(\vec{k})$ and the leaves
    of the tree as $\mathbb{L}(\vec{k})$.
    \item A coupling vector $\vec{p} \in \Z_+^{D+1}$, so that
    \begin{equation}
        0 = p_0 \leq p_1 \leq \cdots \leq p_D = 1\,.
    \end{equation}
\end{enumerate}
Given these parameters, we generate a family of Hamiltonians $(H^{(u)})_{u \in \mathbb{L}(\vec{k})}$ as follows.
Generate independent instances of the Hamiltonian $\widetilde{H}^{(v)}$ for every non-root $v \in \mathbb{T}(\vec{k})$, and scale\footnote{For the spin glass, scale the variance of each Gaussian by $p_d - p_{d-1}$; for the CSP, scale the number of clauses by $p_d - p_{d-1}$.} each by a factor of $p_d - p_{d-1}$ where $d$ is the depth of $v$ in the tree. Use these to construct the Hamiltonians $H^{(u)}$ for each $u  \in \mathbb{L}(\vec{k})$ defined by
\begin{align}
    H^{(u)} = \sum_{d=0}^{D-1} \widetilde{H}^{(a(u,d))}\,,
\end{align}
where $a(u,d)$ is the $d$th ancestor of $u$. 
The \emph{size} of the ensemble is $\abs{\mathbb{L}(\vec{k})} = \prod_{i=1}^D \vec{k}_i$.
\end{definition}

\begin{definition}[Branching OGP]
\label{defn:branchingogp}
    A tree-coupled ensemble of Hamiltonians exhibits a \emph{branching OGP} with gap $\eta$ at value $v$ if:
    \begin{itemize}
        \item there is $\vec{q} \in \Z_+^{D+1}$ with $0 \leq q_0 < q_1 < \cdots < q_D = 1$,
        \item define $Q \in \R^{\mathbb{L}(\vec{k})\times \mathbb{L}(\vec{k})}$ by $Q_{u,v} = q_{lca(u,v)}$, where $lca(u,v)$ is the depth of the least common ancestor of $u$ and $v$,
        \item define 
        \begin{align*}
        \cQ(\eta) =& \left\{\vec{\sigma} \in (\{\pm 1\}^n)^{\mathbb{L}(\vec{k})} : \forall u, v \in \mathbb{L}(\vec{k}). \; \abs{R(\sigma^{(u)}, \sigma^{(v)}) - Q_{u,v}} \leq \eta \right\}
        \end{align*}
        \item it holds that for all $(\sigma^{(u)})_{u \in \mathbb{L}(\vec{k})} \in \cQ(\eta)$, the average value over $u \in \mathbb{L}(\vec{k})$ of $H^{(u)}(\sigma^{(u)})$ is at most $v$.
    \end{itemize}
    The \emph{size} of the OGP is $\abs{\mathbb{L}(\vec{k})} = \prod_{i=1}^D \vec{k}_i$.
\end{definition}

\begin{remark}
\label{remark:branchingogp_modify_defn}
    Our definition is slightly simplified from \cite{huang2021tight}.
    We always fix (using their notation) $\mathbf{m} = 0$,
    because as we show below in \Cref{app:debiasing}, it suffices
    to obstruct algorithms which have mean zero.
\end{remark}

A tree-coupled ensemble is an instance of an $(A,b)$-coupled model, \Cref{defn:coupled-model}, using the following choice of parameters:
\begin{itemize}
    \item $\ell = |\mathbb{L}(\vec{k})|$.
    \item $\ell' = |\mathbb{T}(\vec{k})|$.
    \item $A_{u,v} = 1$ if $u = v$, or $v \in \mathbb{T}(\vec{k})$ is a non-root ancestor of $u \in \mathbb{L}(\vec{k})$; otherwise $A_{u,v} = 0$.
    \item $b_{v} = p_d - p_{d-1}$ for every $v \in \mathbb{T}(\vec{k})$ at the $d$th level of the tree.
\end{itemize}

The presence of the branching OGP is determined by the set
\begin{equation}
S = \{r \in \cR: \forall u,v \in \mathbb{L}(\vec{k}) \mathrel{.} \abs{r_{\{u,v\}} - Q_{u,v}} < \eta\}\,.
\end{equation}

Therefore, by \Cref{prop:S-OGP}, the branching OGP transfers from the spin
glass to the CSP.

\section{Implications for algorithmic hardness}
\label{sec:hardness}
We prove that the OGP has consequences for the performance of many Max-CSP algorithms, specifically those which are \emph{overlap-concentrated}.
This section follows \cite[Section 3]{huang2021tight} which proves a similar result for spin glasses,
with the changes that (1)~an instance of (the Poisson version of) the CSP model consists of i.i.d Poisson random variables whereas a spin glass consists of i.i.d Gaussian random variables; 
(2)~we debias the algorithms first to simplify the proof.
We do not enforce any time or space constraints on algorithms and the proofs apply equally well to arbitrary functions satisfying the assumptions.

\begin{definition}[$t$-correlated instances]\label{defn:t-correlated}
For $t \in [0,1]$, a \emph{$t$-correlated pair of instances} $\cI^1_t$ and $\cI^2_t$ of the model~\csp{} is created by:
\begin{enumerate}
    \item Generate $\Pois(\alpha t n)$ random clauses of the CSP. Let $\cI^1_t$ and $\cI^2_t$ both be initially equal to this instance.
    \item Independently generate two additional sets of $\Pois((1-t)\alpha n)$ clauses. Add the first set of clauses to $\cI^1_t$ and the second set to $\cI^2_t$.
\end{enumerate}
Equivalently, we let $\cI_t^1$ be an instance with $\Pois(\al n)$ clauses,
then independently with probability $1-t$ for each possible clause and sign pattern, we re-sample the multiplicity $\Pois\left(\tfrac{\al}{2^kn^{k-1}}\right)$ of the clause and sign pattern to generate $\cI_t^2$.

Equivalently, sample $\cI^1_t, \cI^2_t$ from the $(A,b)$-coupled model with $A = [[1,1,0],[1,0,1]], b = [t, 1-t,1-t]$.
\end{definition}

\begin{definition}[Overlap-concentrated algorithm, {\cite[Definition 2.1]{huang2021tight}}]\label{defn:overlap-conc}
    A deterministic algorithm $\cA$ is \emph{overlap-concentrated} if
    for every $t \in [0,1]$ and constant $\delta > 0$, the event
    \begin{equation}
    \abs{R(\cA(H_t^1), \cA(H_t^2)) - \E R(\cA(H_t^1), \cA(H_t^2))} \leq \delta
    \end{equation}
    occurs w.h.p, where $H_t^1, H_t^2$ are $t$-correlated instances.
\end{definition}

Our results also apply to randomized algorithms which are overlap-concentrated w.h.p over the random seed.
For the remainder of the section, we assume the algorithm is deterministic.

By a debiasing procedure, we may assume without loss of generality that
an overlap-concentrated algorithm $\cA$ satisfies $\E[\cA(H)] = 0$. See \Cref{app:debiasing} for details.

Fix an algorithm $\cA$.
For $t \in [0,1]$, let $H^1_t$ and $H^2_t$ be a $t$-correlated pair of Hamiltonians for \ \csp{}. We define the \emph{correlation function} $\chi: [0,1] \to \R$ for $\cA$ by
\begin{align}
    \chi(t) = \E\bigg[R(\cA(H^1_t), \cA(H^2_t))\bigg]\, .
\end{align}

The correlation function $\chi$ satisfies the following properties.
\begin{proposition}\label{prop:properties-correlation}
    For any algorithm $\cA$ run on an instance $\cI$ of \csp{} such that $\E [\cA(\cI)] = 0$, $\chi: [0,1] \to [0,1]$ satisfies
    \begin{enumerate}
        \item $\chi$ is continuous,
        \item $\chi(t)$ is strictly increasing,
        \item $\chi(0) = 0, \chi(1) = 1$
        \item for all $t \in [0,1]$, $\chi(t) \leq t$.
    \end{enumerate}
\end{proposition}
\begin{proof}[Proof sketch]
As we show in~\Cref{app:prop-correlation}, the correlation function $\chi(p)$ is equal to the noise stability of $\cA$ viewed as a function of independent Poisson random variables, from which the listed properties easily follow.
See the appendix for the full proof.
\end{proof}

\begin{definition}[Correlation function, {\cite[Proposition 3.1]{huang2021tight}}]\label{defn:corr-function}
    A function $\chi:[0,1] \to [0,1]$ that satisfies the properties
    of \Cref{prop:properties-correlation} is called a \emph{correlation function}.
\end{definition}

\begin{definition}[$\eta$-forbid an algorithm]
    An OGP with gap $(s,t)$ $\eta$-forbids an algorithm $\cA$ if $\E R(\cA(H^1_r), \cA(H^2_r))$ is in the $\eta$-interior of the interval $(s,t)$ for $H^1_r$ and $H^2_r$ a pair of $r$-correlated instances for some $r\in[0,1]$.
    In other words, the open ball with radius $\eta$ around the point is contained in $(s,t)$.

    More generally, an $S$-OGP $\eta$-forbids an algorithm $\cA$ if $\E Q(\cA(H_1), \dots, \cA(H_\el))$ 
    is in the $\eta$-interior of $S$.
\end{definition}

\begin{definition}[$\eta$-forbid a correlation function]
An $S$-OGP $\eta$-forbids a correlation function $\chi$ if $S \subseteq [-1,1]^{\binom{\el}{2}}$ only depends on pairwise
    overlaps, and $(\chi(\rho_{i,j}))_{i,j}$ is in the $\eta$-interior of $S$ where $\rho_{i,j} = \E H_i(\sigma)H_j(\sigma)$ is the correlation of Hamiltonians $H_i$ and $H_j$ (for any fixed arbitrary $\sigma$).
\end{definition}

Observe that if an OGP forbids the correlation function of an algorithm, then it also forbids the algorithm.

\begin{theorem}[Existence of even spin glass OGP, {\cite[Proposition 3.2]{huang2021tight}}]
\label{thm:ogp-existence}
    For all even functions $\xi$ and $v > $ ALG, there is $\eta > 0$ and an integer $\el$ such that
    for all correlation functions $\chi$
    there is a coupled model for \sgxi{} on at most $\el$ Hamiltonians and a set $S$ $\eta$-forbidding $\chi$ such that with high probability the coupled model exhibits an $S$-OGP at value $v$.
    Specifically, the coupled model is a tree-coupled ensemble and the OGP is a branching OGP.
\end{theorem}

We now show that the presence of a branching OGP on a random Max-CSP upper-bounds the performance of any overlap-concentrated algorithm.

\begin{theorem}[Formal version of \Cref{cor:ogp_and_overlap_concentration_obstructed}]
\label{thm:branchingogp_stops_csp}
    Suppose that \csp{} exhibits a size-$\el$ branching OGP at value $v$ for a constant $\el$.
    Suppose that an overlap-concentrated algorithm $\cA$
    is $\eta$-forbidden by the OGP for a constant $\eta > 0$.
    Then w.h.p the output of $\cA$ has value at most $v$.
\end{theorem}

\begin{proof}
    Let $p$ be the probability that $\cA$ outputs a solution of value at least $v$
    when run on a single instance of \csp{}, so that we want to prove $p \leq o_n(1)$.
    Fix the correlation structure of the Hamiltonians for which the branching OGP holds, and let $H^{(u)}$ be the correlated Hamiltonians.
    Consider running algorithm $\cA$ on each of the instances $H^{(u)}$.
    Let $\cE$ be the event that the average value of the output solutions is at least $v$.
    On the one hand, by \cite[Proof of Proposition 3.6(a)]{huang2021tight},
    \begin{equation}\label{eq:cE:lower-bound}
        \Pr[\cE] \geq p^\el\,.
    \end{equation}
    
    On the other hand, $\cE$ is unlikely, because the expected output of $\cA$ lies
    in the forbidden region of the OGP, and if the output of $\cA$
    lies in the forbidden region, then by definition $\cE$ does not hold.
    Specifically,
    \begin{align}
        \Pr[\cE] &\leq \Pr[\cA(H^{(u)})\text{ outside of forbidden region}]\,.
    \end{align}
    The right-hand side can be union-bounded
    across the $\binom{\el}{2}$ pairwise overlaps describing the forbidden region. 
    In order to lie outside the region, the deviation from the expectation must be at least $\eta$ in one of the coordinates.
    By overlap concentration, 
    the deviation probability is at most $\binom{\el}{2}\cdot o_n(1) = o_n(1)$.
    Combining this with \Cref{eq:cE:lower-bound} and taking the $\el$th root yields the claim.
\end{proof}

\begin{remark}
Note that \cite{huang2021tight} obstruct a broader class of algorithms that can output solutions in the convex hull of the solution space; since we  restrict to algorithms that output valid values, we do not use the technical considerations of \cite[Section 3.3]{huang2021tight}.
\end{remark}

Combining \Cref{thm:ogp-existence}, \Cref{prop:S-OGP}, \Cref{thm:branchingogp_stops_csp}, we conclude the following:

\begin{corollary}\label{cor:alg-threshold}
    Let $\cI$ be an instance of \csp{} such that every predicate $f$ in the support of $\Lam$ is even.
    For all $\eps > 0$ there is $\alpha_0 > 0$ such that for all $\al \geq \al_0$ and all overlap-concentrated algorithms, w.h.p it holds that
    \begin{equation}
        \frac{1}{n}H^\al(\cA(\cI)) \leq \E_{f \sim \Lam}[f] + \frac{ALG + \eps}{\sqrt{\alpha}}\,.
    \end{equation}
    where ALG is defined in \Cref{def:alg} for $\xi$ related to $\Lam$ as in \Cref{eq:main-eqn}.
\end{corollary}

Using \Cref{thm:alg-opt}, which states that ALG is strictly less than the ground state energy for even spin glasses with no degree-2 part,
\begin{corollary}[Formal version of \Cref{cor:obstruct-even-csps}]
\label{cor:formal-even-csps}
    Let $\cI$ be an instance of \csp{} such that for every predicate $f$ in the support of $\Lam$, the only nonzero Fourier coefficients of $f$ have even degree $j \ge 4$.
    There are constants $\eps > 0, \alpha_0 > 0$ such that for all $\alpha \geq \alpha_0$ and all overlap-concentrated algorithms $\cA$, w.h.p it holds that
    \begin{equation}
        \frac{1}{n}H^\al(\cA(\cI)) \leq v_\cI - \frac{\eps}{\sqrt{\alpha}}\,.
    \end{equation}
    In particular, for $\al = \al_0$, the value is bounded away from $v_\cI$ by a constant factor.
\end{corollary}

\begin{remark}
\label{rmk:more-concentration}
    As in \Cref{rmk:concentration},
    assuming:
    \begin{enumerate}
        \item the OGP holds with probability at least $1 -\exp(-C_1 \cdot n)$, which is shown by \cite{huang2021tight} in their proof of \Cref{thm:ogp-existence}
        \item $\cA$ is overlap-concentrated with probability at least $1 - \exp(-C_2\cdot n)$ for $C_2 = C_2(\delta)$
    \end{enumerate}
    then the conclusions of \Cref{thm:branchingogp_stops_csp}, \Cref{cor:alg-threshold}, \Cref{cor:formal-even-csps}
    hold with probability at least $1 - \exp(-C\cdot n)$,
    where $C$ depends on $C_1, C_2$.
\end{remark}

\section{Discussion}
\label{sec:discussion}

We establish a formal average-case link between spin glasses and Max-CSPs of certain minimum clause densities.
This link shows an equivalence of optimal value, overlap gaps, and hardness for a large class of algorithms.
Curiously, every Max-CSP with the same noise stability polynomial is linked to the same spin glass. 

As part of this work, we extend the list of Max-CSPs known to have an OGP. 
It is an open question to completely classify which spin glasses (and thus which Max-CSPs) have an OGP.
In the spherical spin glass setting~\cite[Proposition 1]{subag18}, the weight of the quadratic terms exactly determine the presence of an OGP; the same may be true on the hypercube.
There is a technical hurdle to proving the existence of OGPs on spin glasses when the mixture polynomial is not even. For example, the associated spin glass to Max-$3$XOR is not proven to have an OGP, although it is expected to have one~\cite{alaoui2020optimization}.

It is also possible that the onset of the branching OGP of~\cite{huang2021tight} 
marks the hardness threshold for \emph{all} efficient algorithms for spin glasses and Max-CSPs, and not just overlap-concentrated algorithms.
Furthermore, it is possible that a single algorithm (namely, suitably-applied message-passing) is the optimal algorithm.
This is remarkably similar to the situation for worst-case analysis vis-{\`a}-vis the \emph{Unique Games Conjecture (UGC)}~\cite{khot2002power}.
The UGC implies that the standard SDP is the optimal approximation algorithm for any CSP~\cite{raghavendra2008optimal}. It is interesting to investigate whether other algorithms, such as Sum-of-Squares algorithms, can be designed in a way that has equivalent performance to approximate message-passing, although there are existing certification lower bounds for a variety of models~\cite{bhattiprolu2016sum, ghosh2020sum}.
Our curiosity is heightened by the observation that the Parisi formula, and algorithms for spin glasses (and hence
average-case CSPs), only use the ``degree-2'' part of the overlap distribution on $\cR^{(\el)}$, analogous to how the basic SDP is also a ``degree-2'' algorithm.

How general is overlap concentration? \cite{huang2021tight} show that Langevin dynamics and certain families of approximate message-passing algorithms are overlap-concentrated on spin glasses, and therefore obstructed in the presence of a branching OGP to a constant given by the extended Parisi formula.
Algorithms representable as low-degree polynomials are known to be \emph{stable} on spin glasses, and obstructed for some spin glasses~\cite{gamarnik2020low} and CSPs~\cite{bresler2022algorithmic}. 
However, it remains a technical challenge to show that low-degree polynomials are overlap-concentrated:
\begin{conjecture}[Low-degree polynomials are overlap-concentrated on CSPs and spin glasses]
    Any low-degree polynomial algorithm that solves a typical instance of~\,$\csp{}$ or $\sgxi{}$ as defined in~\cite[Definition 2.3]{gamarnik2020low} is overlap-concentrated.
\end{conjecture}

There may be other forms of equivalence between spin glasses and Max-CSPs. We conjecture that the two models have the same \emph{distribution} of overlap vectors:
\begin{conjecture}[Equivalence of distribution of overlap vector, informal]
\label{conj:distributional_equivalence_informal}
Take any Max-CSP and consider the associated spin glass \sgxi{}, where $\xi$ is defined as in \Cref{eq:main-eqn}. For any coupled ensemble of the models, the distribution of $Q(\sigma_1, \dots, \sigma_\el)$ for random near-optimal solutions converges in some sense as $n \to \infty$. Furthermore, for $\alpha \to \infty$ the distribution for the CSP converges to the distribution for the spin glass.
\end{conjecture}
It is also likely that algorithms beyond the QAOA have identical average-case performance on a random instance of Max-CSP and that of the corresponding spin glass. In fact, we suspect that every Max-CSP has an optimal algorithm related to message-passing that obeys this equivalence.

Spin glasses may also be related to more classes of CSPs, such as those with non-Boolean inputs and those without random literal signs. 
Also of interest
is the problem of refuting the CSP when it is unsatisfiable, which typically requires $\alpha$ superconstant.
It is not known how to use statistical physics methods to study refutation~\cite{how_to_refute_random_csp, hsieh2021certifying}.
Another open problem (related to studying CSPs at small constraint density $\alpha$) is to determine the optimal average-case value to higher-order terms in $\alpha$.

\newpage

\section*{Acknowledgements}
We thank Antares Chen for collaborating during the early stages of this project.

We also thank Peter J. Love for helpful feedback on a previous version of this manuscript. 

KM thanks Ryan Robinett for a tip on improving the numerical integration scheme. JSS did some of this work as a visiting student at Bocconi University. CJ thanks Tim Stumpf for help with the probabilistic arguments.

We thank the anonymous reviewers for many suggestions to improve the text,
for pointing out an error in the proof of \Cref{thm:branchingogp_stops_csp}, and for the reference~\cite{barbier2022strong}.

\bibliography{main.bib}
\bibliographystyle{alpha-beta}

\newpage
\appendix

\addtocontents{toc}{\protect\setcounter{tocdepth}{2}}
\appendixpage
\section{Concentration Proofs}
\label{app:csp-technical}

In this section, we furnish the proofs of concentration inequalities
that we use. We start with \Cref{cor:max-concentration}.
The exponent in this corollary can be improved to $\exp\left(-\frac{x^2}{2a}\right)$, which is the statement of the Borell-TIS inequality,
but we utilize this argument since it analogously proves \Cref{cor:max-concentration-csp}.

\maxConcentration*

\begin{proof}
    Take $f(\sigma) = 1$ in \Cref{lem:gaussian-concentration} to obtain the free energy $X = Z_H(\beta)$. By the lemma, we have
    \begin{equation}
        \Pr\left[\abs{\frac{1}{\beta}Z_H(\beta) - \E \frac{1}{\beta} Z_H(\beta)} \geq x\right] \leq 2 \exp\left(\frac{-x^2}{4s^2}\right) \,.
    \end{equation}
    On the other hand, by \Cref{fact:maxvalue}, $\lim_{\beta \to \infty} \frac{1}{\beta} Z_H(\beta) \eqas \max_{\sigma \in \{\pm 1\}^n}H(\sigma)$, and since the right-hand side concentration is uniform in $\beta$, we may take the limit $\beta \to \infty$ on the left-hand side to conclude the claim.
    Formally, a.s. convergence implies $\lim_{\beta \to \infty}\E \frac{1}{\beta}Z_H(\beta) = \E \max_{\sigma \in \{\pm1 \}^n} H(\sigma)$. Slutsky's theorem then implies
    \begin{equation}
        \frac{1}{\beta}Z_H(\beta) - \E \frac{1}{\beta} Z_H(\beta) \toas \max_{\sigma \in \{\pm 1\}^n} H(\sigma) - \E \max_{\sigma \in \{\pm 1\}^n} H(\sigma)\,.
    \end{equation}
    This implies that the probability of lying in the test region $(-\infty, -x] \cup [x, \infty)$ is the same between the right side and the limit of the left-hand side.
\end{proof}

We will now prove \Cref{prop:poisson-concentration}.
This proof applies to the Poisson model. The same proof works for the exact model, because McDiarmid's inequality is proven with a martingale argument, which still holds for non-independent random variables.
We will use the following variant of McDiarmid's inequality,
which applies to a Lipschitz function of i.i.d random variables which are highly biased.

\begin{lemma}[{\cite[Lemma 7.6 (arXiv version)]{chou2022limitations}}]
Suppose that $X_1, \dots X_N$ are sampled i.i.d from a distribution $D$ over a finite set $\cX$ such that $D$ assigns probability $1-p$ to a particular outcome $x_0 \in \cX$.
Let $F: \cX^N \to \R$ satisfy a bounded-differences inequality, so that
\begin{equation}
\abs{F(x_1, \dots, x_{i-1}, x_i, x_{i+1}, \dots, x_N) - F(x_1, \dots, x_{i-1}, x'_i, x_{i+1}, \dots, x_N)}  \leq c
\end{equation}
for all $x_1 ,\dots, x_n, x_i' \in \cX$. Then
\begin{equation}
\Pr[\abs{F(X_1, \dots, X_N) - \E F(X_1, \dots, X_N)} \geq \eps] \leq 2 \exp\left(\frac{-\eps^2}{2Np(2-p)c^2 + 2c\eps / 3}\right)\,.
\end{equation}
\end{lemma}
We deduce the following corollary for Poisson random variables.

\begin{corollary}\label{cor:poisson-concentration}
    Suppose that $X_1, \dots, X_N$ are sampled i.i.d from $\Pois(p)$.
    Let $F(X_1, \dots, X_N)$ satisfy the bounded-differences inequality, so that
    \begin{equation}
        \abs{F(x_1, \dots, x_{i-1}, x_i, x_{i+1}, \dots, x_N) - F(x_1, \dots, x_{i-1}, x_i+1, x_{i+1}, \dots, x_N)}  \leq c
    \end{equation}
    for all $x_1, \dots, x_n \in \N$. Then
    \begin{equation}\label{eq:mcdiarmid}
    \Pr[\abs{F(X_1, \dots, X_N) - \E F(X_1, \dots, X_N)} \geq \eps] \leq 2 \exp\left(\frac{-\eps^2}{2Np(2-p)c^2 + 2c\eps / 3}\right)\,.
    \end{equation}
\end{corollary}
\begin{proof}
    Fix $m \in \N$. We approximate $\Pois(p) \approx \Bin(m, p/m)$,
    apply the lemma to the corresponding (finitely many, 0/1-supported) Bernoulli random variables, then take the limit $m \to \infty$.
    
    Formally, let $(X_{i,j})_{i \in [N], j \in [m]}$ be i.i.d $\Ber(p/m)$, and let
    \begin{equation}
        F'(X_{i,j}) := F(\sum_{j=1}^m X_{1,j}, \dots, \sum_{j=1}^m X_{N,j})\,.
    \end{equation}
    Applying the version of McDiarmid's inequality to $F'$, we conclude
    \begin{equation}
        \Pr[\abs{F'(X_{i,j}) - \E F'(X_{i,j})} \geq \eps] \leq 2 \exp\left(\frac{-\eps^2}{2Np(2-p)c^2 + 2c\eps / 3}\right)\,.
    \end{equation}
    Notice that the right-hand side is independent of $m$.
    Taking the limit $m \to \infty$ independent of all the other parameters, the left-hand side converges to $\Pr[\abs{F(X_1, \dots, X_N) - \E F(X_1, \dots, X_N)} \geq \eps]$ as desired.
    This is justified as follows. We have $\Bin(m, p/m) \overset{d}{\to} \Pois(p)$ as $m \to \infty$. Writing $F'(X_{i,j}) = F(X'_1, \dots, X'_N)$ where $X'_i \sim \Bin(m, p/m)$,
    by the continuous mapping theorem ($F$ is continuous since the domain is discrete) and Slutsky's theorem,
    \[ F(X'_1, \dots, X'_N) - \E F(X'_1, \dots, X'_N) \overset{d}{\to} F(X_1, \dots, X_N) - \E F(X_1, \dots, X_N) \,.\]
    Hence we may conclude that the probability of lying in the test region $(-\infty, -\eps] \cup [\eps, \infty)$ is the same between the right side and the limit of the left-hand side.
\end{proof}

\cspConcentration*

\begin{proof}
We consider the $n^k$ random variables $X_e$ equal to the multiplicity of 
each edge $e \in [n]^k$ and the function $F = \frac{1}{\beta} X$.
The $X_e$ are i.i.d as $\Pois(\alpha / n^{k-1})$.
If a single edge is added, then since $H(\sigma)$ changes by at most 
\begin{equation}
\frac{\max_{f \in \supp(\Lam), x \in \{\pm1\}^k}\abs{f(x)}}{\al} \leq \frac{1}{\al}\,,
\end{equation}
we also have that $F$ changes by only at most $\frac{1}{\al}$.

Therefore we may apply \Cref{cor:poisson-concentration} with constant $c = \frac{1}{{\al}}$, $p = \frac{\alpha}{n^{k-1}}$, and $\eps = \frac{x}{\beta}$.
The denominator in the right-hand side exponent of \Cref{eq:mcdiarmid} is 
\begin{align}
    & \frac{-\eps^2}{2Np(2-p)c^2 + 2c\eps / 3}\\
    = & \frac{-x^2}{\left(\tfrac{2n}{\al}(2-\tfrac{\al}{n^{k-1}}) + \tfrac{2x}{3 \al\beta}\right)\beta^2}\\
    = & \frac{-x^2\al}{\left(2n(2-\tfrac{\al}{n^{k-1}}) + \tfrac{2x}{3\beta}\right)\beta^2}\\
    \leq & \frac{-x^2\al}{4\left(n + \tfrac{x}{\beta}\right)\beta^2}\,.
\end{align}
This completes the proof.
\end{proof}

\poissonVsExact*

\begin{proof}
Let $m$ be the number of edges in the Poisson model and $\alpha n$ in the exact model. Couple the two models so that the first $\al n$ edges of
the Poisson model equal the exact model. Comparing the Hamiltonians
of the two instances, we find
\begin{align}
    \frac{1}{\alpha} \sum_{e \in E\left(\cI^{(exact)}\right)}f_e(\sigma_e) - \frac{1}{\alpha}\cdot\abs{m - \alpha n}&\leq \frac{1}{\alpha} \sum_{e \in E\left(\cI^{(Pois)}\right)}f_e(\sigma_e) \leq \frac{1}{\alpha} \sum_{e \in E\left(\cI^{(exact)}\right)}f_e(\sigma_e) + \frac{1}{\alpha}\cdot \abs{m - \alpha n}\,,
\end{align}
where we have used that $|f(x)| \leq 1$. This gives a corresponding bound on the difference in free energy density:
\begin{equation}\label{eq:pois-vs-exact}
    \abs{\phi^{(\text{Pois})}_{\cH,S} - \phi^{(\text{exact})}_{\cH,S}} \leq \frac{|m - \alpha n|}{\al n}\,.
\end{equation}
Taking the expectation,
\begin{align}
    \abs{\E\phi^{(\text{Pois})}_{\cH,S} - \E\phi^{(\text{exact})}_{\cH,S}} &\leq \E\abs{\phi^{(\text{Pois})}_{\cH,S} - \phi^{(\text{exact})}_{\cH,S}} & (\text{Jensen's inequality})\\
    &\leq \frac{\E_{m \sim \Pois(\al n)}|m - \alpha n|}{\al n} & (\text{\Cref{eq:pois-vs-exact}})\\
    &\leq \frac{\left(\E_{m \sim \Pois(\al n)}(m - \alpha n)^2\right)^{1/2}}{\al n} & (\text{Jensen's inequality})\\
    &= \frac{\sqrt{\alpha n}}{\al n} = \frac{1}{\sqrt{\alpha n}}\,.
\end{align}
\end{proof}

\section{Correlation functions of sparse Hamiltonians}\label{app:prop-correlation}

In this section, we prove \Cref{prop:properties-correlation}.
The proof uses Fourier analysis of $\cA$ in the same way as \cite[Proof of Proposition 3.1]{huang2021tight}. We must use the more general Efron-Stein basis and a mild generalization of the noise operator defined in~\Cref{defn:rho-correlated}.

\begin{definition}[Efron-Stein decomposition]
\label{defn:efron-stein}
For a function $f \in L^2(\Omega^m, \pi^{\otimes m})$, we define $f^{\subseteq S}(x) = \E_{x_{\bar{S}}} [f(x) \mid x_S]$, where the expectation means that the variables $x_i$ for $i \in \bar{S}$ are independently resampled from $\pi$, and $f^{=S}(x) = \sum_{J \subseteq S} (-1)^{|S|-|J|}f^{\subseteq J}(x)$.

A core property of this decomposition, as stated in \cite[Section II.2.2]{globalhypercontractivity}, is that $f^{=S}$ depends only on $x_i$ for $i\in S$, and furthermore $\iprod{f^{=S}, g} = 0$ for any $g$ that does not depend on all $x_i$ for $i \in S$. Additionally, by~\cite[Theorem 8.35 and Proposition 8.36]{o2014analysis},
\begin{align}
    f = \sum_{S \subseteq [m]}f^{=S}\, = \sum_{S \subseteq [m]}\bigg(\sum_{\alpha \in \N^m_{|\Omega|},\,\supp(\alpha) = S}\widehat{f}(\alpha)\phi_\alpha\bigg)\, ,
\end{align}
where $\N^m_{|\Omega|} = \{0,\dots,|\Omega|-1\}^m$ and,
\begin{align}
    \phi_\alpha = \prod_{j=1}^{m}\phi_{\alpha_i}\, ,
\end{align}
with $\phi_0 = \mathsf{Id}$.
\end{definition}
Using this decomposition, any $\cA: \Omega^m \to [-1,1]^n$ under the input measure $\pi^{\otimes m}$ can be written as
\begin{align}
    \cA = (f_1,\dots,f_n)\, ,
\end{align}
where each $f_i$ can be described as in~\Cref{defn:efron-stein}.

We now define the noise operator on a domain $\Omega^m$ with an arbitrary product measure $\pi^{\otimes m}$:
\begin{definition}[Noise operator,~{\cite{globalhypercontractivity, zhao2021generalizations}}]
\label{defn:noise-operator}
Given a probability space $(\Omega, \pi)$ and a noise parameter $\rho \in [0, 1]$, define the \emph{noise operator} $T_{\rho}$, which acts on a function $f \in L_2(\Omega^m,\pi^{\otimes m})$ as
\begin{align}
    T_{\rho}[f](\sigma_1, \dots, \sigma_m) = \E_{\tau \sim N_{\rho}(\sigma)} \bigg[f(\tau_1, \dots, \tau_m)\bigg] \,,
\end{align}
where the distribution $N_{\rho}(\sigma)$
chooses each coordinate $\tau_i$ independently as 
    \begin{align}
       \begin{cases}
        \tau_i = \sigma_i & \text{with probability } \rho\,, \\
        \tau_i \sim \pi & \text{with probability } 1- \rho\,.
        \end{cases}
        \end{align}
\end{definition}

\begin{proof}[Proof of {\Cref{prop:properties-correlation}}]
Define the Fourier weight on a level $j$ as
\begin{align}
    W_j = \frac{1}{n}\sum_{i=1}^n\sum_{S \subseteq [m],|S| = j}\norm{\cA^{= S}_i}^2\, .
\end{align}
Then, when $x, y$ are generated as in~\Cref{defn:t-correlated} with $t = p$,
\begin{align}
    \chi(p) &= \E\left[R(\cA(x), \cA(y))\right] = \frac{1}{n}\E[\langle \cA(x), \cA(y) \rangle] = \frac{1}{n}\E\left[\left\langle \cA(x), T_p[\cA](x)\right\rangle\right] \nonumber \\
    &= \frac{1}{n}\E\left[\left\langle T_{\!\!\sqrt{p}} [\cA](x), T_{\!\!\sqrt{p}} [\cA](x) \right\rangle\right] = \frac{1}{n}\sum_{i=1}^n\sum_{S \subseteq [m]}p^{|S|}\norm{f^{= S}_i}^2 = \sum_{j \geq 0}p^{|S|}W_j\, .
\end{align}
This follows for two reasons. First, the noise operator acts on a $p\cdot m$-sized random subset of $[m]$, equivalent to the (average) number of shared clauses generated by~\Cref{defn:t-correlated}. Second, the noise operator satisfies the semi-group property of $T_p = T_{\sqrt{p}}\circ T_{\sqrt{p}}$.

The desired properties for $\chi(p)$ now follow straightforwardly.
$\chi(p)$ is continuous as it is a polynomial in $p$.
It is strictly increasing as the coefficients are nonnegative and not all zero.
$\chi(0) = 0$ and $\chi(1) = 1$ by computation.  Finally we have
\begin{align}
    \chi(p) = \sum_{j \geq 1}p^{|S|}W_j \leq p \cdot \left(\sum_{j \geq 1} W_j\right) = p\,.
\end{align}
\end{proof}

\section{Debiasing}
\label{app:debiasing}

\begin{proposition}[Debiasing spin glass algorithms]
    Let $\cA$ be a deterministic algorithm for optimizing an instance $H$ of $\sgxi{}$ such that 
    \begin{enumerate}
        \item $\cA$ is overlap-concentrated
        \item $H(\cA(H)) \geq v$ whp.
    \end{enumerate}
    Then there is a deterministic algorithm $\cA'$ such that $\E \cA'(H) = 0$ and
    \begin{enumerate}
        \item $\cA'$ is overlap-concentrated
        \item $H(\cA'(H)) \geq v - o(n)$ whp.
    \end{enumerate}
\end{proposition}
\begin{proof}
    We use the first column of $J^{(2)}$ to symmetrize the algorithm.
    Since the first column is only a small fraction of the input, this will not significantly affect the output value.
    Our technique can be interpreted as a derandomization procedure for randomized spin glass algorithms that use a small number of random bits.
    
    In the case of $\xi(s) = c_1^2 s$, the optimal algorithm is $\cA(H) = \sgn(J^{(1)})$ and we ignore this case.
    
    Let $\widetilde{C} \in \R^n$ be the first column of $J^{(2)}$ and let $C = \sgn(\widetilde{C}) \in \{\pm 1\}^n$ be the entrywise sign.
    Let
    \begin{equation}
    J'^{(p)} := J^{(p)} \circ C^{\otimes p} \quad \text{restricted to coordinates 2 through $n$.}
    \end{equation}
    where $\circ$ is the entrywise/Hadamard product.
    The algorithm $\cA'$ runs $\cA$ on $J'^{(p)}$ to produce a candidate solution $\sigma' \in \{\pm 1\}^{n-1}$ on coordinates 2 through $n$. Then it outputs $C \circ [1, \sigma']$.
    
    To compute the expectation, observe that $C$ and $J'^{(p)}$ are probabilistically independent. Therefore 
    \begin{equation}
    \E[C \circ [1, \sigma']] = \E[C] \circ [1, \E[\sigma']] = 0 \circ [1, \E[\sigma']] = 0\,.
    \end{equation}
    The value of the output is
    \begin{align}
        H(C \circ [1,\sigma']) &= \sum_{p \geq 1} \frac{c_p}{n^{(p-1)/2}} \left\langle{J^{(p)}}, {C^{\otimes p} \circ [1,\sigma']^{\otimes p}}\right\rangle\\
        &= \sum_{p \geq 1} \frac{c_p}{n^{(p-1)/2}} \left(\left\langle{J'^{(p)}}, { \sigma'^{\otimes p}}\right\rangle + \sum_{\substack{(i_1,\dots, i_p) \in [n]^p:\\\text{at least one equals 1}}} J^{(p)}[i_1, \dots, i_p]\sigma'_{i_1}\cdots \sigma'_{i_p}\right)
    \end{align}
    Since $\sigma'$ is independent of $J^{(p)}[i_1, \dots, i_p]$ such that at least one index equals 1, the second term is a mean-zero Gaussian
    with variance $O(n^{p-1})$.
    Whp,
    \begin{align}
        &= \sum_{p \geq 1} \frac{c_p}{n^{(p-1)/2}} \left(\left\langle{J'^{(p)}}, { \sigma'^{\otimes p}}\right\rangle \pm O(n^{(p-1)/2}\sqrt{\log n})\right)\\
        &= \sum_{p \geq 1} \frac{c_p}{(n-1)^{(p-1)/2}} \left\langle{J'^{(p)}}, {\sigma'^{\otimes p}}\right\rangle  \pm  o(n) \\
        &\geq v - o(n)\,.
    \end{align}
    where the last inequality holds whp by assumption on $\cA$.
    
    It remains to show that $\cA'$ remains overlap-concentrated.
    Consider $t$-correlated Hamiltonians $H_t^1,H_t^2$.
    Let $C^1, C^2 \in \{\pm1\}^n$ be the signs of the respective first columns of $J^{(2)}$ and let $q = \E[C^1_1C^2_1]$ (note that all of the coordinates have the same distribution). We have
    \begin{align}
        \E[ R(\cA'(H_t^1), \cA'(H_t^2))]
        &= q\cdot \E [R(\sigma'^1, \sigma'^2)]\,.
    \end{align}
    We also have
    \begin{align}
        R(\cA'(H_t^1), \cA'(H_t^2)) &= \frac{1}{n}C_1^1C_1^2 + \frac{1}{n}\sum_{i = 2}^n \sigma'^1_i\sigma'^2_i C^1_i C^2_i\,.
    \end{align}
    As before, $\sigma'$ and $C^1,C^2$ are independent. With $\sigma'$ fixed, 
    the terms $\sigma'^1_i\sigma'^2_i C_i^1C_i^2$ are independent with deviation from the mean bounded by 2.
    Their mean over $C^1, C^2$ is $q \cdot R(\sigma'^1, \sigma'^2)$.
    Therefore by Bernstein's inequality,
    \begin{align*}
        \Pr_{C^1, C^2}\left[\abs{R(\cA'(H_t^1), \cA'(H_t^2)) - q\cdot R(\sigma'^1, \sigma'^2)} \geq \frac{L}{n}\right] \leq \exp\left(\frac{-L^2/2}{n + 2L/3}\right)\,.
    \end{align*}
    By overlap concentration, 
    \begin{equation}
        \Pr\left[\abs{q \cdot R(\sigma'^1, \sigma'^2) -  q \cdot \E[R(\sigma'^1, \sigma'^2)]} \geq \delta\right] \leq o_n(1)\,.
    \end{equation}
    Take $\frac{L}{n} = \delta$ and apply the triangle inequality and a union bound to conclude the claim.
\end{proof}

\begin{proposition}[Debiasing random CSP algorithms]
\label{prop:debias-csp}
    The same statement holds for the model $\csp{}$.
\end{proposition}
\begin{proof}
    In the CSP model, we may use the first $\tfrac{n}{\log n}$ constraints to symmetrize the algorithm.
    Since the entropy of a single constraint is approximately $k \log n$,
    this will provide at least $n$ bits of entropy for symmetrization.\footnote{This is called \emph{randomness extraction} in computer science literature.}
    Any $\omega(1)$ bits of entropy are enough to maintain overlap concentration whp.

    Let $\cI$ be the instance of $\csp{}$.
    Let $\widetilde{C}$ be the first $\frac{n}{\log n}$ clauses of the CSP and let $C \in \{\pm 1\}^n$ extract the bit labels of the first variable in each clause, so that $C$ equals $n$ independent random bits.
    Let $\cI'$ be the CSP instance consisting of
    the remaining clauses of $\cI$,
    then flipping the sign of the variable $i$ by $C_i$ in each
    constraint.
    
    The algorithm $\cA'$ runs $\cA$ on $\cI'$ to produce a candidate solution $\sigma' \in \{\pm 1\}^{n}$. Then it outputs $C \circ\sig'$.

    The analysis proceeds as in the previous proof.
    Since $C$ and $\cI'$ are independent,
    \begin{equation}
        \E[C \circ \sig'] = \E[C] \circ \E[\sig'] = 0 \circ \E[\sig'] = 0\,.
    \end{equation}

    The number of clauses used in $\widetilde{C}$ is $o(n)$, hence the Hamiltonian value is affected by at most $o(n)$.

    Overlap concentration of $\cA'$ follows exactly as in the prior proof.
\end{proof}

\end{document}